\documentclass[journal]{IEEEtran}

\usepackage{bbm}
\usepackage{bm}
\usepackage{amsmath}
\usepackage{amsthm}
\usepackage{amssymb}
\usepackage{mathrsfs}
\usepackage{optidef}
\usepackage{algorithm}
\usepackage[noend]{algpseudocode}
\usepackage{cite}
\usepackage{graphicx}
\usepackage{textcomp}
\usepackage{xcolor}
\usepackage{caption}
\usepackage{subcaption}
\usepackage{hyperref}
\usepackage{booktabs}
\usepackage{multirow}
\usepackage{amsfonts}
\usepackage{algorithm}
\usepackage{array}
\usepackage{stfloats}
\usepackage{url}
\usepackage{verbatim}
\usepackage{booktabs}
\usepackage{tabularx}

\begin{document}
% Theorem Family
\newtheorem{theorem}{Theorem}
\newtheorem{lemma}{Lemma}
\newtheorem{proposition}{Proposition}
\newtheorem{corollary}{Corollary}
\newtheorem{definition}{Definition}
\newtheorem{assumption}{Assumption}
\newtheorem{example}{Example}
\newtheorem{remark}{Remark}
\newtheorem{conjecture}{Conjecture}

% Big Fraction
\newcommand\ddfrac[2]{\frac{\displaystyle #1}{\displaystyle #2}}

\title{Kalman Filter Aided Federated Koopman Learning}

\author{Yutao Chen,~\IEEEmembership{Member, IEEE}, Wei Chen,~\IEEEmembership{Senior Member, IEEE}
    \thanks{This work is supported in part by the NSFC/RGC Joint Research Scheme under Grant No. 62261160390/N\_HKUST656/22 and in part by the National Natural Science Foundation of China under Grant No. 62471276. (\emph{Corresponding author: Wei Chen}.)
    }
    \thanks{Yutao Chen and Wei Chen are with the Department of Electronic Engineering, Tsinghua University, Beijing 100084, China, and also with the State Key Laboratory of Space Network and Communications, as well as, the Beijing National Research Center for Information Science and Technology (email: cheny1995@tsinghua.edu.cn, wchen@tsinghua.edu.cn).}
}

\maketitle

\begin{abstract}
    Real-time control and estimation are pivotal for applications such as industrial automation and future healthcare. The realization of this vision relies heavily on efficient interactions with nonlinear systems. Therefore, Koopman learning, which leverages the power of deep learning to linearize nonlinear systems, has been one of the most successful examples of mitigating the complexity inherent in nonlinearity. However, the existing literature assumes access to accurate system states and abundant high-quality data for Koopman analysis, which is usually impractical in real-world scenarios. To fill this void, this paper considers the case where only observations of the system are available and where the observation data is insufficient to accomplish an independent Koopman analysis. To this end, we propose Kalman Filter aided Federated Koopman Learning (KF-FedKL), which pioneers the combination of Kalman filtering and federated learning with Koopman analysis. By doing so, we can achieve collaborative linearization with privacy guarantees. Specifically, we employ a straightforward yet efficient loss function to drive the training of a deep Koopman network for linearization. To obtain system information devoid of individual information from observation data, we leverage the unscented Kalman filter and the unscented Rauch-Tung-Striebel smoother. To achieve collaboration between clients, we adopt the federated learning framework and develop a modified FedAvg algorithm to orchestrate the collaboration. A convergence analysis of the proposed framework is also presented. Finally, through extensive numerical simulations, we showcase the performance of KF-FedKL under various situations.
\end{abstract}


\section{Introduction}
% Nonlinear system
The rapid development of the Internet of Things (IoT)~\cite{IOT2014, IOT2023} has led to the growing importance of efficient interactions with real-world physical systems, which are often described as nonlinear systems. These interactions span various domains, ranging from cyber-physical systems (CPS)~\cite{cyber2008} to haptic communications~\cite{Haptic2012}. To achieve efficient interactions, nonlinearity becomes the most significant obstacle. Therefore, the linearization of nonlinear systems has attracted significant attention~\cite{Linearization,mauroy2020koopman,koopman1}, and the Koopman operator theory~\cite{Koopman1931} has emerged as a promising theory. The research in Koopman analysis has primarily remained theoretical until the introduction of Koopman learning~\cite{NNKoopman2018, DeepKoopman2018}, which leverages the power of deep learning to achieve linearization. The primary advantages of deep learning-based approaches lie in the fact that they do not require explicit knowledge of the nonlinear system, and the framework is easily adaptable to various nonlinear systems.

% Federated Learning
Nevertheless, these approaches also face limitations, such as the need for sufficient and comprehensive data. To compensate for this type of limitation, data sharing is a straightforward solution. However, directly sharing raw data may raise security and privacy concerns. Therefore, we propose the integration of Koopman learning with Federated Learning (FL)~\cite{FedAve2017} to mitigate the adverse effects of unsatisfactory data quality and quantity. Enthusiasm for FL has persisted since it was proposed~\cite{FL2019,FLAsynchronous2019,FLPrivacy2017,FLLinear2016}, with research focusing on communication and optimization schemes, data security and privacy enhancements, and applications across diverse learning paradigms. In this paper, we adopt horizontal FL, where a central server orchestrates communications under a synchronous optimization scheme. An autoencoder-based neural network is trained to achieve linearization, and a model aggregation method is adopted to preserve privacy. For network training, we propose a loss function specifically designed for training under real-time data acquisition.

% Motivation and Use Case
The federated learning of the Koopman operator is motivated by various real-world problems, particularly in clinical applications. For instance, cardiac motion, which exhibits nonlinear dynamics, has significant physiological and clinical implications~\cite{tong2007sampled}. However, individual medical research institutions often struggle to collect sufficient data and may encounter biases due to geographic, ethnic, and environmental factors. To acquire diverse and ample data, collaboration is an efficient approach. Nevertheless, raw data sharing is often restricted due to privacy regulations and institutional confidentiality. In this context, federated learning has emerged as an effective solution~\cite{medicine4}. Through the collaborative linearization facilitated by the proposed framework, institutions can construct comprehensive models collaboratively while preserving the privacy of raw measurements and associated metadata.

% Kalman Filtering
Another essential assumption in the existing literature is the access to accurate system states from numerical solutions. However, this assumption often diverges from reality, where only observation data about the system is accessible. Although direct utilization of observation data for Koopman learning is feasible, the results will be confined solely to specific observation mechanisms. In response to this obstacle, we leverage the Kalman filter to estimate system states from observation data. The combination of the Koopman analysis and the Kalman filter is explored differently in~\cite{KKF2016, KKF2023} where the Koopman operator theory is used as a linearization tool to enable the application of the classical Kalman filter on nonlinear systems. The authors in~\cite{KalmanNet} present a Kalman filter based estimator for non-linear systems with partial information. In this paper, the Kalman filter is employed to estimate the system states from observation data, thereby effectively reducing the amount of individual-specific information embedded in the raw observations, such as the information embedded in the observation methods and capabilities. As a result, when the estimated states are used for linearization, the results contain limited individual information, which facilitates the implementation of federated learning and preserves client privacy when reporting the results to the central server.

% Main contribution
The main contributions and novelty of this paper are summarized as follows.
\begin{itemize}
    \item This paper presents the first attempt to combine federated learning and Koopman learning to address the adverse effects of data inadequacy while providing privacy guarantees. Through this pioneering effort, we expand the scope of applications for Koopman learning.
    \item We introduce the Kalman filter to estimate the system states from observation data, as opposed to the existing literature where the system's full states are obtained from numerical solutions and/or physical sensor measurements. The introduction of the Kalman filter enhances privacy protection and enables collaboration via federated learning since the estimates contain minimal individual information compared to the raw observation data. Through this approach, we further enhance the applicability of our results.
    \item We propose the Deep Koopman Network (DKN), which can linearize not only the dynamics of the system state but also a time-delay embedded representation of the system states. Moreover, we propose a loss function that penalizes only one-step prediction errors, thereby enhancing its practicality when data is collected in real-time.
    \item This paper differs from the existing literature in that we use the estimated system states as training data. In this case, it is crucial to fine-tune the network size and training policy to mitigate the impact of estimation errors. This paper offers valuable insights into the fine-tuning through comprehensive numerical simulations.
    \item We consider the case where the observation data arrives in a stochastic manner rather than being readily available. The stochastic nature of the arrival introduces additional challenges in optimizing the training policy. The impact of different training policies is explored numerically via extensive simulations.
\end{itemize}

% Paper structure
The remainder of this paper is organized as follows. Section~\ref{sec:SystemModel} details the system model. Section~\ref{sec:KalmanKoopman} introduces the Kalman filter for obtaining system information devoid of individual information and the Koopman operator theory for system linearization. In Section~\ref{sec:KFKL}, we present the Deep Koopman Network (DKN) for linearization, propose the framework for facilitating the federated learning of the Koopman operator, and provide a theoretical convergence analysis. The paper concludes with Section~\ref{sec:Numerical}, where extensive numerical results are presented and analyzed to showcase the performance.

\section{System Model}\label{sec:SystemModel}
We consider a scenario where multiple clients observe a common nonlinear system and aim to linearize it. The clients obtain system information through observations only. At the same time, the clients collaborate with the help of a central server to address the data insufficiency caused by adverse observation conditions. An illustration of the system model is provided in Fig.~\ref{fig:SystemModel}.
\begin{figure}[t]
    \centering
    \includegraphics[width=\columnwidth]{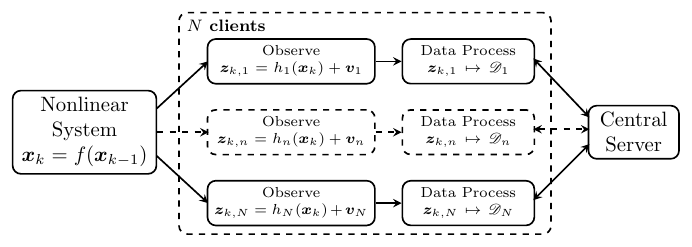}
    \caption{An illustration of the multi-client system model, where $\mathscr{D}_n$ is the processed data that is exchanged between client $n$ and the central server.}
    \label{fig:SystemModel}
\end{figure}

We investigate a system with $N$ clients and a central server. In the system, each client independently observes a common $d_x$-dimensional nonlinear system by sampling at a fixed rate of $1/\nu$. The sampling rate is assumed to be sufficiently high. Let $\bm{x}_k\in\mathbbm{R}^{d_x}$ denote the system state at time $t=k\nu$. Then, the dynamics of the nonlinear system are described by
\begin{equation}
    \bm{x}_{k} = f(\bm{x}_{k-1}),
\end{equation}
where $f:\mathbbm{R}^{d_x}\rightarrow\mathbbm{R}^{d_x}$ represents the nonlinear dynamics. However, obtaining an accurate mathematical description of a nonlinear system is challenging and also beyond the scope of this paper. To account for this difficulty, we introduce an independent random noise $\bm{s}_n\in\mathbbm{R}^{d_x}$ for each client to represent the uncertainty about the evolution of the system. More precisely, we have
\begin{equation}\label{eq:noisy_state}
    \bm{x}_{k} = f_n(\bm{x}_{k-1}) + \bm{s}_n,
\end{equation}
where $f_n$ is the estimated system dynamics at client $n$. We assume $\bm{s}_{n}$ follows the multivariate normal distribution $\mathcal{N}(\bm{0},\bm{\Sigma}_{f,n})$ and is independent across all clients. Moreover, we consider the case where the clients only partially observe the system state. Hence, we use a $d_z$-dimensional vector $\bm{z}_{k}\in\mathbbm{R}^{d_z}$ to denote the observed state corresponding to $\bm{x}_k$. The relationship between the observed state and the system state is captured by $h:\mathbbm{R}^{d_x}\rightarrow\mathbbm{R}^{d_z}$. The clients can adopt different $h$ to account for the variations in observation mechanisms. We also introduce a random noise $\bm{v}\in\mathbbm{R}^{d_z}$ to represent the uncertainty in the observations. Mathematically, we have
\begin{equation}
    \bm{z}_{k,n} = h_n(\bm{x}_{k}) + \bm{v}_n,
\end{equation}
where the subscript $n$ distinguishes clients. We assume that $\bm{v}_n$ follows the multivariate normal distribution $\mathcal{N}(\bm{0},\bm{\Sigma}_{h,n})$ and is independent across all clients. Note that the clients have no access to the system state $\bm{x}_k$. For the central server, it plays a dual role in the system. First, as a coordinator, the central server manages the collaboration among clients. Second, as a trusted third party, the clients can securely communicate with the central server. The central server operates discretely, meaning its internal time is slotted, and each action can only occur at the beginning of a time slot.

In many applications such as drone control, celestial dynamics analysis, and smart grid monitoring, linear representations of nonlinear systems are desirable for efficient analysis, control, and remote monitoring. However, traditional linearization methods often fail to capture global system behaviors, while data-driven approaches typically require large and diverse datasets. To address these limitations, collaborative learning among multiple clients offers a promising solution. A straightforward approach is direct data sharing. However, this raises significant privacy concerns. Therefore, this paper proposes a framework that enables collaborative global linearization in a privacy-preserving manner. To this end, we address the following key challenges. The first challenge is to achieve a linearization that minimizes the inclusion of client-specific information, such as that embedded in the observation function $h_n(\cdot)$ and the noise terms $\bm{s}_n$ and $\bm{v}_n$. Direct linearization based on local observation data inherently encodes client-specific information and the results are valid only for the originating client. Hence, we need to ensure that the linearization does not rely on client-specific information. The second challenge lies in enabling effective collaboration. Collaborative training is essential for improving generalization by leveraging more diverse data across clients. Although direct data sharing is a straightforward method for enabling collaboration, it compromises data privacy. Consequently, the central server must facilitate collaboration without compromising the privacy of individual clients.

\section{Kalman Filter Assisted Koopman learning}\label{sec:KalmanKoopman}
In this section, we focus on the first challenge outlined above. To achieve linearization that contains limited client-specific information, the initial step is to estimate the system states from the observation data. For this purpose, we leverage the Kalman filter's robust capability in estimating system states from noisy observations. Then, we employ the Koopman learning for linearization using the estimated system states.

\subsection{Kalman Filter for State Estimation}\label{sec:Kalman}
The Kalman filter effectively estimates the system states from observations in a way that minimizes the mean squared error. In this paper, we adopt the Unscented Kalman Filter~\cite{UKF1997} (UKF), which is an extension of the classical Kalman filter to nonlinear systems.

The cornerstone of UKF is the Unscented Transformation (UT), which is a method for approximating the statistical properties of a nonlinearly transformed random variable. For simplicity, we temporarily drop the subscript $n$ used to distinguish between clients. The key concept in UT is the sigma vectors. Specifically, for state $\bm{x}_k$, we construct a matrix $\mathscr{X}_k$ containing $2d_x+1$ sigma vectors. Let $[\cdot]_j$ denote the $j$th column of a matrix. The sigma vector $[\mathscr{X}_k]_j$ is defined as
\begin{equation}\label{eq:SigmaVector}
    \begin{split}
        [\mathscr{X}_{k}]_0       & \triangleq \bm{x}_{k},                                                                  \\
        [\mathscr{X}_{k}]_j       & \triangleq \bm{x}_{k} + \left[\sqrt{(d_x+\lambda)\bm{P}_k}\right]_j,\quad j=1,\dots,d_x \\
        [\mathscr{X}_{k}]_{j+d_x} & \triangleq \bm{x}_{k} - \left[\sqrt{(d_x+\lambda)\bm{P}_k}\right]_j,\quad j=1,\dots,d_x
    \end{split}
\end{equation}
where $\bm{P}_k\triangleq\mathbbm{E}\left[\left(\bm{x}_k-\mathbbm{E}\left[\bm{x}_k\right]\right)\left(\bm{x}_k-\mathbbm{E}\left[\bm{x}_k\right]\right)^T\right]$ is the covariance matrix and $\lambda\triangleq\alpha^2(d_x+\kappa)-d_x$. Here, $\alpha$ is a hyperparameter that controls the distribution of the sigma vectors around $\bm{x}_k$ and $\kappa$ is a scaling hyperparameter. We also define
\begin{equation}
    \begin{split}
        W_0^m         & \triangleq \frac{\lambda}{d_x+\lambda},                    \\
        W_0^c         & \triangleq \frac{\lambda}{d_x+\lambda} + 1-\alpha^2+\beta, \\
        W_j^m = W_j^c & \triangleq \frac{1}{2(d_x+\lambda)},\quad j=1,\dots,2d_x
    \end{split}
\end{equation}
where $\beta$ is a hyperparameter introduced to integrate prior knowledge of the distribution of $\bm{x}_k$~\cite{UKF2001}. The discussion on the choices of $\alpha$, $\kappa$, and $\beta$ is omitted here, and please refer to~\cite{UKF2001} for details.

With the sigma vectors and the weights established, we can introduce the UKF, which is encapsulated in the following two steps.
\paragraph{Prediction step} In this step, the sigma points $\mathscr{X}_{k-1}$, generated using the previous estimates $\hat{\bm{x}}_{k-1}$ and $\hat{\bm{P}}_{k-1}$, are used to predict the current state and the corresponding covariance. To this end, $\mathscr{X}_{k-1}$ is first propagated through $f$, which yields
\begin{equation}\label{eq:SigmaPredict}
    [\mathscr{X}^-_{k}]_j = f([\mathscr{X}_{k-1}]_j).\quad j=0,1,\dots,2d_x
\end{equation}
Then, $\bm{x}_k^-$ is calculated as the weighted mean of the sigma vectors, and $\bm{P}^-_k$ is the sum of $\bm{\Sigma}_f$ and the weighted covariance of the sigma vectors. Specifically, we have
\begin{equation}\label{eq:KalmanPredict}
    \begin{split}
        \bm{x}_k^- & = \sum_{j=0}^{2d_x}W_j^m[\mathscr{X}_k^-]_j,                                                                \\
        \bm{P}^-_k & = \sum_{j=0}^{2d_x}W_j^c([\mathscr{X}^-_k]_j-\bm{x}^-_k)([\mathscr{X}^-_k]_j-\bm{x}^-_k)^T + \bm{\Sigma}_f.
    \end{split}
\end{equation}
At this point, we have obtained initial predictions based on the previous estimates and knowledge of the system dynamics.
\paragraph{Correction step} In this step, the predicted state $\bm{x}_k^-$ and the corresponding covariance $\bm{P}^-_k$ are corrected using the observed state $\bm{z}_{k}$. To this end, $\mathscr{X}^-_{k}$ is first propagated through $h$, which results in
\begin{equation}\label{eq:SigmaObserve}
    [\mathscr{Z}^-_{k}]_j = h([\mathscr{X}^-_{k}]_j).\quad j=0,1,\dots,2d_x
\end{equation}
Then, the predicted observation, denoted as $\bm{z}^-_{k}$, and the corresponding covariance, denoted as $\bm{P}_{\bm{z},\bm{z}}$, are calculated by following a similar procedure as outlined in~\eqref{eq:KalmanPredict}. Specifically, we have
\begin{equation}\label{eq:ObservationPredict}
    \begin{split}
        \bm{z}^-_{k}           & = \sum_{j=0}^{2d_x}W^m_j[\mathscr{Z}^-_{k}]_j,                                                                      \\
        \bm{P}_{\bm{z},\bm{z}} & = \sum_{j=0}^{2d_x}W_j^c([\mathscr{Z}^-_{k}]_j-\bm{z}^-_{k})([\mathscr{Z}^-_{k}]_j-\bm{z}^-_{k})^T + \bm{\Sigma}_h.
    \end{split}
\end{equation}
Moreover, we define $\bm{P}_{\bm{x},\bm{z}}$ as the weighted cross covariance between $\bm{x}_k^-$ and $\bm{z}_{k}^-$. Specifically, we have
\begin{equation}\label{eq:CrossVariance}
    \bm{P}_{\bm{x},\bm{z}} \triangleq \sum_{j=0}^{2d_x}W_j^c([\mathscr{X}^-_k]_j-\bm{x}^-_k)([\mathscr{Z}^-_{k}]_j-\bm{z}^-_{k})^T.
\end{equation}
Then, the Kalman gain $\mathcal{K}_k$, which modulates the estimate's reliance on the observed state versus the system model, is computed as
\begin{equation}\label{eq:KalmanGain}
    \mathcal{K}_k = \bm{P}_{\bm{x},\bm{z}}\bm{P}^{-1}_{\bm{z},\bm{z}}.
\end{equation}
Finally, the estimated system state $\hat{\bm{x}}_k$ and covariance $\hat{\bm{P}}_k$ are corrected as
\begin{equation}\label{eq:KalmanCorrect}
    \begin{split}
        \hat{\bm{x}}_k & = \bm{x}^-_k + \mathcal{K}_k(\bm{z}_{k}-\bm{z}^-_{k}),             \\
        \hat{\bm{P}}_k & = \bm{P}^-_k - \mathcal{K}_k\bm{P}_{\bm{z},\bm{z}}\mathcal{K}^T_k.
    \end{split}
\end{equation}
Let $\hat{\bm{x}}_0$ and $\hat{\bm{P}}_0$ be the initial estimates of the system state and the covariance, respectively. The UKF is summarized in Algorithm~\ref{algo:UKF}.
\algrenewcommand\algorithmicindent{1em}
\begin{algorithm}[!t]
    \begin{algorithmic}[1]
        \Procedure{UKF}{$\hat{\bm{x}}_0$, $\hat{\bm{P}}_0$, $f$, $\bm{\Sigma}_f$, $h$, $\bm{\Sigma}_h$}
        \For{$k\in\{1,\dots,\infty\}$}
        \State Get a new observation $\bm{z}_{k}$.
        \State $\mathscr{X}_{k-1}\leftarrow(\hat{\bm{x}}_{k-1}, \hat{\bm{P}}_{k-1})$ using~\eqref{eq:SigmaVector}.
        \State $(\mathscr{X}_k^-, \bm{x}_k^-, \bm{P}_k^-)\leftarrow(\mathscr{X}_{k-1}, f, \bm{\Sigma}_f)$ using~\eqref{eq:SigmaPredict},~\eqref{eq:KalmanPredict}.
        \State $(\mathscr{Z}_{k}^-, \bm{z}_{k}^-, \bm{P}_{\bm{z},\bm{z}})\leftarrow(\mathscr{X}^-_{k}, h, \bm{\Sigma}_h)$ using~\eqref{eq:SigmaObserve},~\eqref{eq:ObservationPredict}.
        \State $\bm{P}_{\bm{x},\bm{z}}\leftarrow(\mathscr{X}_{k}^-, \bm{x}_k^-, \mathscr{Z}_{k}^-, \bm{z}_{k}^-)$ using~\eqref{eq:CrossVariance}.
        \State $\mathcal{K}_{k}\leftarrow(\bm{P}_{\bm{x},\bm{z}}, \bm{P}^{-1}_{\bm{z},\bm{z}})$ using~\eqref{eq:KalmanGain}.
        \State $(\hat{\bm{x}}_k, \hat{\bm{P}}_k)\leftarrow(\bm{x}_k^-, \bm{P}_k^-, \mathcal{K}_k, \bm{z}_{k}, \bm{z}_{k}^-, \bm{P}_{\bm{z},\bm{z}})$ using~\eqref{eq:KalmanCorrect}.
        \EndFor
        \State \Return $(\hat{\bm{x}}_k, \hat{\bm{P}}_k)$ for $k\in\{1,\dots,\infty\}$
        \EndProcedure
    \end{algorithmic}
    \caption{Unscented Kalman Filter}
    \label{algo:UKF}
\end{algorithm}
\begin{remark}
    In this paper, we adopt the UKF without the resampling step. The resampling operation introduces extra computational overhead and, as noted in~\cite{liu2010unscented,menegaz2015systematization}, is unnecessary under the assumption of additive noise and sufficiently accurate system model. However, when the system model exhibits significant inaccuracies or involves non-additive noise, a resampling mechanism can be readily incorporated to enhance estimation performance. In such cases, at the start of the correction step, a new set of sigma points is generated from the results of the prediction step and used in subsequent calculations. A numerical comparison of these approaches is presented in Section~\ref{sec:Numerical}.
\end{remark}

We notice that UKF leverages only the information up to the current time slot to make estimations for the next time slot. However, in the context of learning, it is common to accumulate a certain amount of data before starting the learning. This provides a compelling reason to improve the state estimation by incorporating future information.

Hence, we adopt the Unscented RTS (URTS) smoother~\cite{Smoother2008}, which involves the following two steps.
\setcounter{paragraph}{0}   % Reset the paragraph counter
\paragraph{Backward prediction step} This step starts with utilizing $\hat{\bm{x}}_k$ and $\hat{\bm{P}}_k$ from the output of UKF to predict the next state $\bm{x}^+_{k+1}$ and the corresponding covariance $\bm{P}^+_{k+1}$. To this end, we first construct $\mathscr{X}_k$ from $\hat{\bm{x}}_k$ and $\hat{\bm{P}}_k$ according to~\eqref{eq:SigmaVector}. Then, we apply $f$ on $\mathscr{X}_{k}$, which yields
\begin{equation}\label{eq:SmootherPredictSigma}
    [\mathscr{X}^+_{k+1}]_j = f([\mathscr{X}_{k}]_j).\quad j=0,1,\dots,2d_x
\end{equation}
Finally, the predicted state and covariance are calculated as
\begin{equation}\label{eq:SmootherPredict}
    \begin{split}
        \bm{x}^+_{k+1}          = & \sum_{j=0}^{2d_x}W_j^m[\mathscr{X}^+_{k+1}]_j,                                                                            \\
        \bm{P}^+_{k+1}          = & \sum_{j=0}^{2d_x}W_j^c([\mathscr{X}^+_{k+1}]_j-\bm{x}^+_{k+1})([\mathscr{X}^+_{k+1}]_j-\bm{x}^+_{k+1})^T + \bm{\Sigma}_f.
    \end{split}
\end{equation}
This step is very similar to the prediction step in UKF.
\paragraph{Backward correction step} In this step, the predicted state and covariance from the backward prediction step are corrected using future states. To this end, we first calculate the cross covariance between $\hat{\bm{x}}_k$ and $\bm{x}_{k+1}^+$ as
\begin{equation}\label{eq:SmootherCrossVariance}
    \bm{P}_{\hat{\bm{x}}_k,\bm{x}_{k+1}^+} = \sum_{j=0}^{2d_x}W_j^c([\mathscr{X}_{k}]_j-\hat{\bm{x}}_{k})([\mathscr{X}^+_{k+1}]_j-\bm{x}^+_{k+1})^T.
\end{equation}
Then, the smoother gain is given by
\begin{equation}\label{eq:SmootherGain}
    \mathcal{G}_k = \bm{P}_{\hat{\bm{x}}_k,\bm{x}_{k+1}^+}[\bm{P}^+_{k+1}]^{-1},
\end{equation}
where $[\cdot]^{-1}$ is the matrix inversion. Finally, we introduce the future information $(\tilde{\bm{x}}_{k+1}, \tilde{\bm{P}}_{k+1})$ to correct the predictions. More precisely, we have
\begin{equation}\label{eq:SmootherCorrect}
    \begin{split}
        \tilde{\bm{x}}_{k} & = \hat{\bm{x}}_k + \mathcal{G}_k(\tilde{\bm{x}}_{k+1}-\bm{x}^+_{k+1}),                \\
        \tilde{\bm{P}}_k   & = \hat{\bm{P}}_k + \mathcal{G}_k(\tilde{\bm{P}}_{k+1}-\bm{P}^+_{k+1})\mathcal{G}_k^T.
    \end{split}
\end{equation}
We assume that the UKF outputs $\zeta$ consecutive estimates. Meanwhile, we let $\tilde{\bm{x}}_{\zeta} = \hat{\bm{x}}_{\zeta}$ and $\tilde{\bm{P}}_{\zeta} = \hat{\bm{P}}_{\zeta}$. Then, the URTS smoother is summarized in Algorithm~\ref{algo:URTSSmoother}.
\begin{algorithm}[!t]
    \begin{algorithmic}[1]
        \Procedure{URTS Smoother}{$\tilde{\bm{x}}_{\zeta}$, $\tilde{\bm{P}}_{\zeta}$, $f$, $\bm{\Sigma}_f$}
        \For{$k\in\{\zeta-1,\dots,0\}$}
        \State $(\hat{\bm{x}}_{k}, \hat{\bm{P}}_k)\leftarrow$ Algorithm~\ref{algo:UKF}.
        \State $\mathscr{X}_k\leftarrow(\hat{\bm{x}}_{k}, \hat{\bm{P}}_{k},)$ using~\eqref{eq:SigmaVector}.
        \State $\mathscr{X}_{k+1}^+\leftarrow(\mathscr{X}_{k}, f)$ using~\eqref{eq:SmootherPredictSigma}.
        \State $(\bm{x}_{k+1}^+, \bm{P}_{k+1}^+)\leftarrow(\mathscr{X}_{k+1}^+, \bm{\Sigma}_f)$ using~\eqref{eq:SmootherPredict}.
        \State $\bm{P}_{\hat{\bm{x}}_k,\bm{x}_{k+1}^+}\leftarrow(\mathscr{X}_{k}, \hat{\bm{x}}_k, \mathscr{X}_{k+1}^+, \bm{x}_{k+1}^+)$ using~\eqref{eq:SmootherCrossVariance}.
        \State $\mathcal{G}_{k}\leftarrow(\bm{P}_{\hat{\bm{x}}_k,\bm{x}_{k+1}^+}, \bm{P}_{k+1}^+)$ using~\eqref{eq:SmootherGain}.
        \State $(\tilde{\bm{x}}_k, \tilde{\bm{P}}_k)\leftarrow(\hat{\bm{x}}_k, \hat{\bm{P}}_k, \mathcal{G}_k, \tilde{\bm{x}}_{k+1}, \bm{x}^+_{k+1}, \tilde{\bm{P}}_{k+1}, \bm{P}_{k+1}^+)$ using~\eqref{eq:SmootherCorrect}.
        \EndFor
        \State \Return $(\tilde{\bm{x}}_k, \tilde{\bm{P}}_k)$ for $k\in\{0,\dots,\zeta-1\}$
        \EndProcedure
    \end{algorithmic}
    \caption{Unscented Rauch-Tung-Striebel Smoother}
    \label{algo:URTSSmoother}
\end{algorithm}

Equipped with UKF and URTS smoother, clients can estimate the system states from the observation data.
\begin{remark}
    The choice of state estimation method can be adapted based on performance requirements and system conditions. Depending on the specific needs, the UKF can be replaced by other Kalman filter variants~\cite{chui2017kalman} or particle filter~\cite{djuric2003particle}. When only the system structure is known, unknown parameters can be incorporated into an augmented state vector and estimated jointly. Representative methods include the unified filter~\cite{bao2023unified} and the learning-based KalmanNet~\cite{KalmanNet}. When neither structure nor parameters are known, model-free approaches, such as the Data-Driven Nonlinear State Estimation (DANSE)~\cite{ghosh2024danse}, become necessary. Note that the proposed framework is compatible with a variety of state estimation methods, provided that these methods can deliver satisfactory state estimations from the observation data.
\end{remark}

\subsection{Linearization via Koopman Operator}\label{sec:KoopmanOperatorTheory}
Koopman operator theory~\cite{Koopman1931} demonstrates the possibility of capturing the dynamics of a nonlinear system by an infinite-dimensional linear operator acting on a Hilbert space of system state functions. Specifically, we define $g:\mathbbm{R}^{d_x}\rightarrow\mathbbm{R}$ as the observable function, which spans an infinite-dimensional Hilbert space $\mathscr{H}$. Then, we introduce a linear operator $\mathcal{K}:\mathscr{H}\rightarrow\mathscr{H}$ with respect to the observable function $g$, which satisfies
\begin{equation}
    \mathcal{K}g(\bm{x}_k) = g(\bm{x}_{k+1}),
\end{equation}
where $\mathcal{K}$ is called the Koopman operator. However, the infinite-dimensionality of $\mathscr{H}$ poses challenges in applied Koopman analysis.

To overcome this, researchers seek to approximate the linear dynamics within a subspace spanned by a finite set of base observable functions $\{g_1,\dots,g_d\}$ that remains invariant under $\mathcal{K}$. More precisely, for all observable functions such that
\begin{equation}
    g(\bm{x}_k) = a_1g_1(\bm{x}_k) + a_2g_2(\bm{x}_k) + \cdots + a_dg_d(\bm{x}_k),
\end{equation}
where $a_i\in\mathbbm{R}$ for $1\leq i\leq d$, the result of the operation of $\mathcal{K}$ remains in the subspace. Mathematically, we have
\begin{equation}
    \mathcal{K}g(\bm{x}_k) = b_1g_1(\bm{x}_k) + b_2g_2(\bm{x}_k) + \cdots + b_dg_d(\bm{x}_k),
\end{equation}
where $b_i\in\mathbbm{R}$ for $1\leq i\leq d$. Such subspace is called the Koopman-invariant subspace. Consequently, a finite-dimensional matrix approximation of $\mathcal{K}$ can be obtained by restricting it to the Koopman-invariant subspace. At the same time, any finite set of the eigenfunctions of $\mathcal{K}$ spans such an invariant subspace~\cite{Koopman2021}. Hence, we can choose the base observable functions as the eigenfunctions $\{\varphi_i\}_{i=1}^{\mathcal{O}}$ of $\mathcal{K}$.

Finally, a linearization of the nonlinear system can be expressed as
\begin{equation}
    \varphi(\bm{x}_{k+1}) = \bm{K}\varphi(\bm{x}_k),
\end{equation}
where $\bm{K}\in\mathbbm{R}^{\mathcal{O}\times \mathcal{O}}$ is a finite-dimensional matrix approximation of $\mathcal{K}$ and $\varphi$ is the $\mathcal{O}$-dimensional vector whose elements are the eigenfunctions of $\mathcal{K}$.

Once $\bm{K}$ and $\varphi$ are obtained, we can map the state in the original space through $\varphi$ to the state in the Koopman-invariant subspace, where the linear dynamics are captured by $\bm{K}$. Subsequently, the state in this subspace can be mapped back to the original space using the inverse of $\varphi$, denoted as $\varphi^{-1}$. However, identifying $\bm{K}$ and $(\varphi,\varphi^{-1})$ theoretically for a general nonlinear system is challenging, even for data-driven approaches like DMD~\cite{FirstDMD}. Therefore, we adopt a deep learning-based approach, where the goal is to train a neural network to learn the $\bm{K}$ and $(\varphi,\varphi^{-1})$. We embrace this approach for two primary reasons. First, the deep learning-based approach requires minimal prior knowledge of the underlying system. Second, such an approach facilitates the implementation of federated learning, which enables collaboration with privacy guarantees.

\section{Federated Learning of Koopman Operator}\label{sec:KFKL}
To achieve collaborative linearization, we first introduce the Deep Koopman Network (DKN), which learns the $\bm{K}$ and $(\varphi,\varphi^{-1})$ introduced in Section~\ref{sec:KoopmanOperatorTheory}. Then, we propose a framework wherein federated learning is implemented to facilitate collaboration.

\subsection{Deep Koopman Network}\label{sec:DeepKoopman}
The matrix approximation $\bm{K}$ of the Koopman operator and the corresponding $\varphi$ should satisfy
\begin{subequations}\label{eq:KoopmanMain}
    \begin{align}
        \bm{K}\varphi(\bm{x}_k) & =\varphi(\bm{x}_{k+1}), \label{eq:KoopmanLinear}                           \\
        \bm{x}_k                & =               \varphi^{-1}(\varphi(\bm{x}_k)). \label{eq:KoopmanMapping}
    \end{align}
\end{subequations}
Equation~\eqref{eq:KoopmanLinear} captures the ability to map the states in the original space to the Koopman-invariant subspace with linear dynamics. Meanwhile, equation~\eqref{eq:KoopmanMapping} represents the ability to map the state in the Koopman-invariant subspace back to the original space. Given the complexity of the theoretical analysis of~\eqref{eq:KoopmanMain}, we adopt a deep learning-based approach. Specifically, we train a neural network to achieve the desired functionalities. To this end, we first notice that the functionalities of $\varphi$ and the corresponding inversion $\varphi^{-1}$ are similar to those of an autoencoder. Consequently, we adopt the autoencoder-based Deep Koopman Network (DKN) depicted in Fig.~\ref{fig:DKNStructure}. In the following, we elaborate on each component of DKN.
\begin{figure}[t]
    \centering
    \includegraphics[width=\columnwidth]{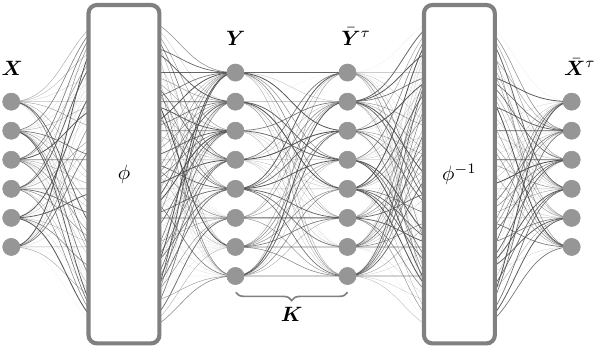}
    \caption{An illustration of DKN, where $\bm{X}$ and $\bar{\bm{X}}^{\tau}$ are the input and output, respectively. $\bm{Y}$ and $\bar{\bm{Y}}^{\tau}$ are the corresponding encoded states.}
    \label{fig:DKNStructure}
\end{figure}
\paragraph{Encoder $\phi$}
The encoder $\phi$ aims to mimic the functionality of $\varphi$, which maps the states in the original space to the Koopman-invariant subspace. To this end, we adopt a feedforward neural network with $\mathcal{Z}$ fully connected layers, each of size $\mathcal{H}$, and incorporate the rectified linear unit (ReLU) as the nonlinear activation. The input to $\phi$ consists of the states corresponding to $\mu$ consecutive time slots. More precisely, the input $\bm{X}$ is given by
\begin{equation}
    \bm{X} = \left[\bm{x}_{k-\mu+1},\bm{x}_{k-\mu+2},\dots,\bm{x}_k\right],
\end{equation}
where $\bm{x}_k$ is the state used in training. It is important to highlight that when $\mu\geq2$, the resulting $\bm{K}$ is not equivalent to the matrix approximation of the Koopman operator in~\eqref{eq:KoopmanMain}. In this case, DKN linearizes the nonlinear dynamics of a time-delay embedded representation of the system states. Nevertheless, $\bm{K}$ remains a valid matrix approximation of the Koopman operator for another system where $\bm{X}$ is the system state. That being said, our discussion will only focus on the case of $\mu=1$, which is widely considered in the literature. Hence, with a slight abuse of notation, we will continue to use $\bm{K}$ in the remainder of this paper. The dimension of the output $\bm{Y}$ aligns with the dimension of $\bm{K}$. For the rest of this paper, we denote the operation of the encoder as $\bm{Y} = \phi(\bm{X})$.
\begin{remark}
    Time-delay embedding offers two key advantages. First, it captures temporal information that is essential in many scenarios. For example, it is useful when the full system state is only partially observed, as supported by Takens' embedding theorem~\cite{takens2006detecting}. Second, the embedded representation naturally serves as data-driven observables in operator-theoretic methods such as Koopman analysis, facilitating system analysis through linearization~\cite{brunton2017chaos,kamb2020time}.
\end{remark}

\paragraph{Koopman operator approximation $\bm{K}$}
The matrix approximation $\bm{K}$ of the Koopman operator is responsible for advancing the state in the Koopman-invariant subspace. In DKN, $\bm{K}$ is a square matrix of dimension $\mathcal{O}$. During training, the elements of $\bm{K}$ are updated to achieve the desired functionality.

\paragraph{Decoder $\phi^{-1}$}
The decoder $\phi^{-1}$, which corresponds to the encoder $\phi$, acts as a tool to map the encoded states back to the original state space. To this end, $\phi^{-1}$ consists of the same hidden layers and nonlinear activation as in the encoder $\phi$. The input to $\phi^{-1}$ is $\bar{\bm{Y}}^{\tau}$, which is the result of the operation of $\bm{K}$ on $\bm{Y}$. Then, the predicted states are given by the output $\bar{\bm{X}}^{\tau}$. Specifically,
\begin{equation}
    \bar{\bm{X}}^{\tau} = \left[\bar{\bm{x}}_{k-\mu+1+\tau},\bar{\bm{x}}_{k-\mu+2+\tau},\dots,\bar{\bm{x}}_{k+\tau}\right],
\end{equation}
where $\tau$ is the prediction depth. We allow the variations in the value of $\tau$ to correspond to the case of different input sizes $\mu$. Nevertheless, as before, our discussion will only focus on the case of $\tau = 1$. The label $\bm{X}^{\tau}$ is defined as
\begin{equation}
    \bm{X}^{\tau} = \left[\bm{x}_{k-\mu+1+\tau},\bm{x}_{k-\mu+2+\tau},\dots,\bm{x}_{k+\tau}\right].
\end{equation}
Then, the difference between $\bar{\bm{X}}^{\tau}$ and $\bm{X}^{\tau}$ is quantified by the loss function to drive the training. Similar to the encoder, we denote the operation of the decoder as $\bar{\bm{X}}^{\tau} = \phi^{-1}(\bar{\bm{Y}}^{\tau})$ for the rest of this paper.

\paragraph{Loss function $\ell$}
The loss function serves as the driving force of DKN. For further elaboration, we first define $||\cdot||_2^2$ as the squared $\ell^2$-norm. We also define $d_{\bm{X}}$ and $d_{\bm{Y}}$ as the dimensionality of $\bm{X}$ and $\bm{Y}$, respectively. Then, building upon~\eqref{eq:KoopmanMain}, the loss function is defined as the weighted sum of the following three loss terms.
\begin{enumerate}
    \item The \textit{linear dynamics loss $\ell_1$} is derived from~\eqref{eq:KoopmanLinear} to facilitate the finding of a Koopman-invariant subspace with linear dynamics. To this end, we define
          \begin{equation}
              \ell_1 \triangleq \frac{1}{d_{\bm{Y}}}||\bar{\bm{Y}}^{\tau}-\bm{Y}^{\tau}||_2^2,
          \end{equation}
          where $\bm{Y}^{\tau}\triangleq\phi(\bm{X}^{\tau})$ is the encoded target state. We recall that $\bar{\bm{Y}}^{\tau} = \bm{K}\phi(\bm{X}) = \bm{K}\bm{Y}$, where $\bm{Y}$ is the encoded input. Hence, $\ell_1$ quantifies the accuracy of $\bm{K}$ in capturing the linear dynamics.
    \item The \textit{reconstruction loss $\ell_2$} is formulated based on~\eqref{eq:KoopmanMapping} to evaluate the performance of the autoencoder. To this end, we define
          \begin{equation}
              \ell_2 \triangleq \frac{1}{d_{\bm{X}}}||\bm{X}-\bar{\bm{X}}||_2^2,
          \end{equation}
          where $\bar{\bm{X}}\triangleq\phi^{-1}(\phi(\bm{X}))$ is the reconstructed input state when the encoder $\phi$ and the decoder $\phi^{-1}$ is successively applied. Hence, $\ell_2$ quantifies the ability of the autoencoder to reconstruct the input.
    \item The \textit{prediction loss $\ell_3$} evaluates the capability to reconstruct the states in the original space from the predicted states in the Koopman-invariant subspace. To this end, we define
          \begin{equation}
              \ell_3\triangleq \frac{1}{d_{\bm{X}}}||\bar{\bm{X}}^{\tau}-\bm{X}^{\tau}||_2^2.
          \end{equation}
          Hence, $\ell_3$ quantifies the overall accuracy when the Koopman operator approximation $\bm{K}$, the encoder $\phi$ and the decoder $\phi^{-1}$ are combined. We include $\ell_3$ because the equalities in~\eqref{eq:KoopmanMain} will not be achieved due to the characteristics of the neural network and the fact that $\bm{K}$ is a finite-dimensional approximation of the infinite-dimensional Koopman operator. Therefore, the presence of $\ell_3$ ensures that the minor errors introduced by the neural network are not amplified to the point where the correspondence between the states in the Koopman-invariant subspace and those in the original space deviates significantly.
\end{enumerate}
With the three loss functions in mind, the composite loss function used in training is defined as a weighted sum of the individual loss functions. More precisely, the composite loss function is defined as
\begin{equation}\label{eq:CompositeLoss}
    \ell\triangleq w_1\ell_1 + w_2\ell_2 + w_3\ell_3.
\end{equation}
The choice of the weight $w_i$ depends on the specific application and the desired behavior of the DKN. Note that each of the three loss functions serves a distinct purpose. To obtain a physically meaningful DKN, it is essential to incorporate all three loss functions. A physically meaningful DKN yields a faithful approximation of the Koopman operator and the associated observables for the underlying nonlinear system, which enables numerous practical advantages, as demonstrated in~\cite{brunton2016koopman,abraham2019active,mezic2013analysis}.
\begin{remark}
    While more sophisticated loss functions have been proposed in the literature, including~\cite{DeepKoopman2018,girgis2022predictive,shi2022deep}, we choose the simple loss function~\eqref{eq:CompositeLoss} for two reasons. First, as an example, we want to show that even with a simple loss function, the proposed framework can effectively achieve its goals. Second, and more importantly, the loss function~\eqref{eq:CompositeLoss} is designed to penalize only the one-step prediction error, rather than multistep prediction errors. This formulation enhances practicality in scenarios involving real-time data collection. As demonstrated by the simulation results presented in Section~\ref{sec:Numerical}, this simple loss function continues to deliver satisfactory performance.
\end{remark}

Once the loss function is determined, the gradients can be calculated. Then, an optimization algorithm can be used to backpropagate the gradients and update the parameters. Ideally, with appropriate hyperparameters, DKN will exhibit the desired behavior after an acceptable number of epochs. However, this satisfactory behavior relies on the availability of sufficient and comprehensive data. One simple approach to tackle this issue is to include extra data from other clients, but this may raise privacy concerns. Therefore, we propose a framework based on federated learning, which enables collaborative training among multiple clients without sharing the raw data.

\subsection{Collaborative Koopman Operator Learning}
In this subsection, we propose Kalman Filter aided Federated Koopman Learning (KF-FedKL), as illustrated in Fig.~\ref{fig:Framework}, where multiple clients independently observe a common nonlinear system and collaboratively train the DKN to achieve linearization.
\begin{figure*}[t]
    \centering
    \includegraphics[width=\textwidth]{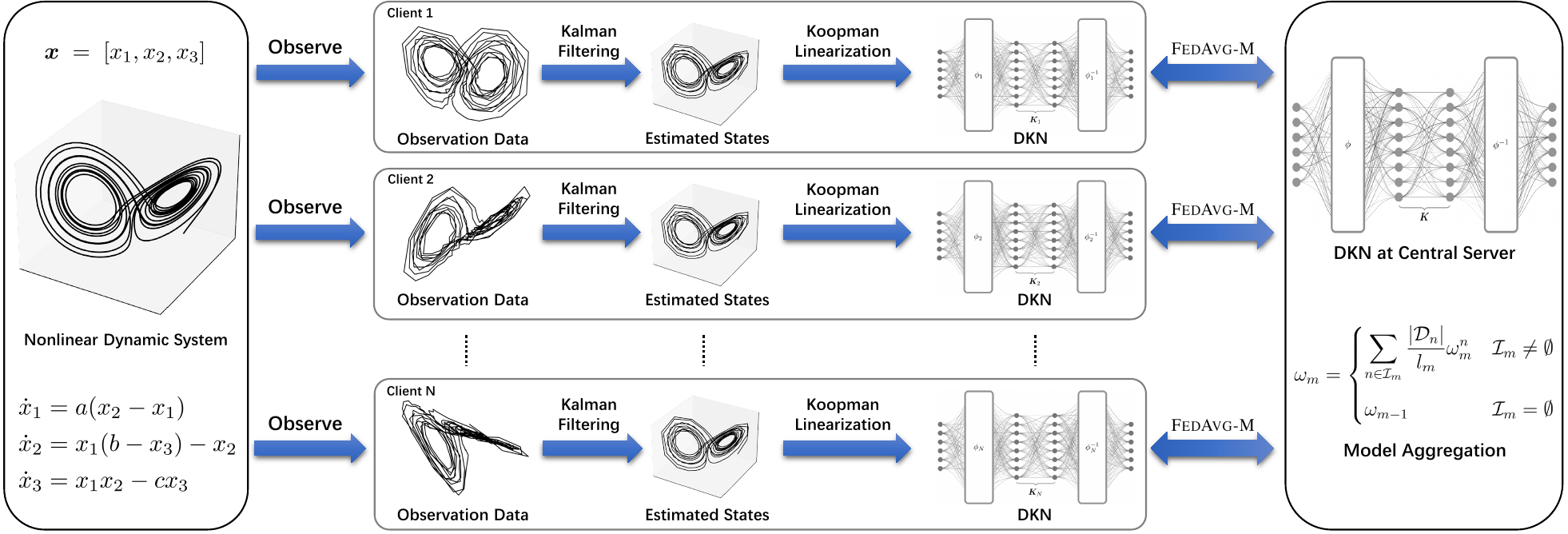}
    \caption{Illustration of the KF-FedKL structure using the Lorenz63 system as an example. Each client independently observes the nonlinear system and estimates the system state using a Kalman filter based on local observations. Clients then train local DKNs using the estimated states and collaborate under the coordination of the central server following \textsc{FedAvg-M}.}
    \label{fig:Framework}
\end{figure*}

To this end, we consider a case where time is slotted and normalized. The time slot is indexed by $m$. To better align with real-world scenarios, we assume that each client may encounter circumstances that cause the observation to fail. Mathematically, we assume that client $n$ has a probability $p_n$ of successful observation in each time slot, and a successful observation yields $\zeta$ consecutive observation data. It is important to note that $p_n$ is assumed to be constant across time slots and independent across clients. Then, based on the available data, the central server activates a subset of clients in each time slot to participate in each training round. For the active clients, the UKF and URTS smoother introduced in Section~\ref{sec:Kalman} are utilized to estimate the system state. These smoothed estimates are then used to train the latest DKN received from the central server. We denote the dataset used for training at client $n$ at time slot $m$ by $\mathcal{X}^n_m$. In KF-FedKL, $\mathcal{X}^n_m$ consists of the estimated states from the Kalman filter and changes over time. Let $\mathcal{I}_{m}$ denote the set of active clients at time slot $m$. Then, the optimization problem for the active client $n\in\mathcal{I}_m$ at time slot $m$ can be formulated as follows.
\begin{argmini}|s|
    {\omega_m^n}{\sum_{\bm{x}\in\mathcal{X}^n_m}\ell(\omega_{m}^n;\bm{x}),}{\label{eq-new30}}{}
\end{argmini}
where $\omega_m^n$ denotes the DKN trained by client $n$ at time slot $m$ and $\ell(\omega_m^n;\bm{x})$ is the loss incurred by $\omega_m^{n}$ on data $\bm{x}$. The optimization is performed by training the model, initialized with $\omega_{m-1}$, on the local dataset $\mathcal{X}^n_m$, where $\omega_{m-1}$ is the central server's DKN at time slot $m-1$. Specifically, the active client $n$ first partitions the data, denoted by $\mathcal{D}_n$, into batches $\mathcal{B}\triangleq[\bm{b}_1,\dots]$. Then, with learning rate $\eta_n$, client $n$ updates the DKN for $E$ epochs. Mathematically, in epoch $e$ and for batch $\bm{b}$, the update is given by
\begin{equation}
    \omega^n_e = \omega^n_{e-1} - \eta_{e-1}\nabla\mathcal{L}_{e-1}^n(\omega_{e-1}^n;\bm{b}),
\end{equation}
where $\omega^n_e$ is the DKN at epoch $e$ and $\mathcal{L}_{e-1}^n(\omega_{e-1}^n;\bm{b})\triangleq\frac{1}{|\bm{b}|}\sum_{\bm{x}\in\bm{b}}\ell(\omega_{e-1}^n;\bm{x})$. Upon completion of local training, the data $\mathcal{D}_n$ is emptied and the active client uploads the result to the central server, where contributions are aggregated to update the DKN at the central server. In KF-FedKL, we adopt the synchronous optimization scheme, in which the central server waits until all the active clients have completed their local training before updating the DKN. We assume that the training time for all clients does not exceed one time slot. The time required for data transmission and DKN updating is ignored.

In the sequel, we present an overview of the algorithm adopted by the central server to manage the collaboration. The proposed algorithm is derived from the \textsc{FedAvg} originally presented in~\cite{FedAve2017}. For clarity, we refer to this modified version as \textsc{FedAvg-M}. In addition, we define $\mathcal{D}\triangleq[\mathcal{D}_1,\dots,\mathcal{D}_N]$. Then, under \textsc{FedAvg-M}, the central server activates clients based on $\mathcal{D}$ and a predetermined policy $\mathcal{P}$ in each time slot. For example, if a threshold-based policy is adopted, the central server activates the client only if the available data exceeds a certain threshold. The central server initiates the training only if there is at least one active client, i.e., $\mathcal{I}_{m}\neq\emptyset$. Otherwise, $\omega_{m} = \omega_{m-1}$. In the case of $\mathcal{I}_{m}\neq\emptyset$, the central server first distributes the latest DKN $\omega_{m-1}$ to each active client. Then, the active clients update the DKN with their own data. Upon completing the local training, the central server receives the updated DKNs from the active clients and adopts a weighted sum of the received DKNs as the new DKN $\omega_m$. The weight assigned to active client $n$ is defined as $\frac{|\mathcal{D}_n|}{l_{m}}$, where $l_{m}\triangleq\sum_{n\in\mathcal{I}_{m}}|\mathcal{D}_n|$ is the sum of the sizes of data available to all active clients in time slot $m$. Mathematically, the new DKN is given by
\begin{equation}\label{eq:WeightedSum}
    \omega_{m} \triangleq \sum_{n\in\mathcal{I}_{m}}\frac{|\mathcal{D}_n|}{l_{m}}\omega_{m}^n,
\end{equation}
where $\omega_m^n$ is the DKN trained on client $n$ at time slot $m$. The pseudocode for \textsc{FedAvg-M} is given in Algorithm~\ref{algo:MyFedAve}.
\begin{algorithm}[!t]
    \begin{algorithmic}[1]
        \Procedure{\textsc{FedAvg-M}}{$\eta_1,\dots,\eta_N$}
        \State Initialize $\omega_0$
        \For{$m\in\{1,2,\dots,M\}$}
        \State $(\mathcal{I}_{m}, l_{m})\leftarrow\mathcal{P}(\mathcal{D})$
        \For{$n\in\mathcal{I}_{m}$ \textbf{in parallel}}
        \State $\omega_{m-1}^n\leftarrow\omega_{m-1}$
        \State $\omega^n_{m}\leftarrow$ \textsc{ClientUpdate}($\omega_{m-1}^n$, $\mathcal{D}_n$, $\eta_n$)
        \State $\mathcal{D}_n\leftarrow\emptyset$
        \EndFor
        \State $\omega_{m}\leftarrow\sum_{n\in\mathcal{I}_{m}}\frac{|\mathcal{D}_n|}{l_{m}}\omega^n_{m} + \mathbbm{1}_{\{\mathcal{I}_{m}=\emptyset\}}\omega_{m-1}$
        \EndFor
        \State \Return $\omega_M$
        \EndProcedure
        \Procedure{ClientUpdate}{$\omega$, $\mathcal{D}$, $\eta$}
        \State Divide $\mathcal{D}$ into batches $\mathcal{B}$; $\omega_0 = \omega$
        \For{$e\in\{1,2,\dots,E\}$ and $\bm{b}\in\mathcal{B}$}
        \State $\omega_e\leftarrow\omega_{e-1}-\eta_{e-1}\nabla\mathcal{L}_{e-1}(\omega_{e-1};\bm{b})$
        \EndFor
        \State \Return $\omega_E$
        \EndProcedure
    \end{algorithmic}
    \caption{Modified \textsc{FedAvg}}
    \label{algo:MyFedAve}
\end{algorithm}
\begin{remark}
    We will not discuss the optimization of the policy $\mathcal{P}$. Different policies will be examined and compared numerically in the next section. At the same time, optimizing the result aggregation strategy is beyond the scope of this paper. We assume that the central server uses the weighted sum approach given in~\eqref{eq:WeightedSum}.
\end{remark}

\subsection{Convergence Analysis}
To facilitate the analysis, we consider the case where the probability of success observation $p_n=p$ for $1\leq n\leq N$. Additionally, the central server adopts the policy under which all the clients with data are selected in each time slot. We divide a single time slot, denoted by the subscript $m$, into smaller time slots, denoted by the subscript $t$, based on the number of training epochs $E$. The relationship between $m$, $t$, and $E$ is illustrated in Fig.~\ref{fig:TimeRelation}. To better distinguish, we refer to time slot $m$ as training round $m$ for the remainder of this subsection.
\begin{figure}[t]
    \centering
    \includegraphics[width=\columnwidth]{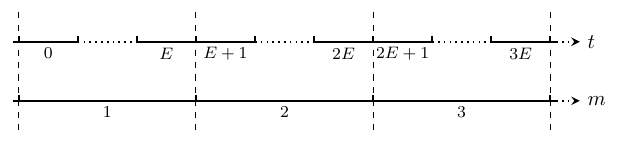}
    \caption{The relationship between $t$, $m$, and $E$.}
    \label{fig:TimeRelation}
\end{figure}
To simplify the presentation of theoretical analysis, we consider the following equivalent scenario. In each time slot $t$, perfect data is available to each client. In each training round $m$, the central server selects each client independently with probability $p$, and the selected client's data will incur an error, representing the error introduced by the Kalman filter. Finally, at the end of each training round $m$, the data at all the clients will be discarded. To simplify the analysis, we assume that the error affects the training in the following way.
\begin{equation}\label{eq:nuDefinition}
    \nu_{t+1}^{n} = \omega^n_t - \eta_t(\nabla\mathcal{L}_t^n(\omega_t^n;\bm{x}) + e_t^n),
\end{equation}
where $\nu_{t+1}^{n}$ is the immediate result of one-step stochastic gradient descent (SGD)~\cite{robbins1951stochastic} update, $\omega^n_t$ is the DKN of client $n$ at time slot $t$, $\bm{x}$ is the data used in update, and $e_t^n$ represents the error in the gradient caused by the error in the data. We assume $e_t^n$ has zero mean with variance $\sigma^2_{n}$, and denote the maximum variance as $\sigma^2=\max\{\sigma^2_n\}$.

To start the convergence analysis, we define $\mathcal{L}_t^n(\omega)\triangleq\frac{1}{|\mathcal{X}_t^n|}\sum_{\bm{x}\in\mathcal{X}_t^n}\ell(\omega;\bm{x})$, where $\mathcal{X}_t^n$ is the data at client $n$ at time slot $t$. Then, the following commonly adopted~\cite{FLConvergenceAnalysis0,FLConvergenceAnalysis1,FLConvergenceAnalysis2} assumptions are made.
\begin{assumption}[L-smooth]\label{ass:1}
    For all $\omega$ and $\nu$, $\mathcal{L}_t^n(\nu)\leq \mathcal{L}_t^n(\omega) + (\nu-\omega)^T\nabla \mathcal{L}_t^n(\omega) + \frac{L}{2}||\nu-\omega||^2_2$ holds for $1\leq n\leq N$ at time slot $t\geq0$.
\end{assumption}
\begin{assumption}[$\mu$-strongly convex]\label{ass:2}
    For all $\omega$ and $\nu$, $\mathcal{L}_t^n(\nu)\geq \mathcal{L}_t^n(\omega) + (\nu-\omega)^T\nabla \mathcal{L}_t^n(\omega) + \frac{\mu}{2}||\nu-\omega||^2_2$ holds for $1\leq n\leq N$ at time slot $t\geq0$.
\end{assumption}
\begin{assumption}[Bounded variance of stochastic gradients]\label{ass:3}
    Let $\xi_t^n$ be sampled from client $n$'s local data at time slot $t$. Then, $\mathbbm{E}\left[||\nabla \mathcal{L}_t^n(\omega_t^n;\xi_t^n)-\nabla \mathcal{L}_t^n(w_t^n)||^2_2\right]\leq \delta_n^2$ holds for $1\leq n\leq N$ at time slot $t\geq0$.
\end{assumption}
\begin{assumption}[Bounded expected squared norm of stochastic gradients]\label{ass:4}
    $\mathbbm{E}\left[||\nabla\mathcal{L}_t^n(\omega_t^n;\xi_t^n)||^2_2\right]\leq G^2$ holds for $1\leq n\leq N$ at time slot $t\geq0$.
\end{assumption}
Let $\mathcal{L}_t(\omega) = \frac{1}{N}\sum_{n=1}^N\mathcal{L}^n_t(\omega)$, and denote the minimum values of $\mathcal{L}^n_t(\omega)$ and $\mathcal{L}_t(\omega)$ for $t\geq0$ as $\mathcal{L}^{n,*}(\omega)$ and $\mathcal{L}^*(\omega)$, respectively. We then present the main results of the convergence analysis.
\begin{theorem}\label{thm:ConvergenceAnalysis}
    When assumptions~\ref{ass:1} to~\ref{ass:4} hold, $\eta_t$ is decreasing, $\eta_t\leq \frac{1}{4L}$, and $\eta_t \leq 2\eta_{t+E}$ for $t\geq0$, $\omega_{t+1}$ satisfies the following recursive form.
    \begin{equation}
        \begin{split}
            \mathbbm{E}\big[||\omega_{t+1}- & \omega^*||_2^2\big]\leq (1-\mu\eta_t)\mathbbm{E}\left[||\omega_t -\omega^*||_2^2\right] +                      \\
                                            & \eta_t^2\left[8(E-1)^2G^2 + 6L\Gamma\right] +  \frac{1}{N^2}\sum_{n=1}^N\left(\delta_n^2 + \sigma_n^2\right) + \\
                                            & \sum_{k=1}^NC_kE^2\eta_{t}^2\left(G^2+\sigma^2\right),
        \end{split}
    \end{equation}
    where $\omega^*$ is the target DKN and
    \begin{equation}
        C_k\triangleq\frac{4(N-k)N!p^k(1-p)^{N-k}}{k(N-1)k!(N-k)!},
    \end{equation}
    \begin{equation}
        \Gamma = \mathcal{L}^* - \frac{1}{N}\sum_{n=1}^N\mathcal{L}^{n,*}.
    \end{equation}
\end{theorem}
\begin{proof}
    The proof follows the methodology presented in~\cite{FLConvergenceAnalysis0}. Two major differences in our analysis are that, under our policy, the number of clients selected each time is a random variable, and the effect of erroneous data is considered. The complete proof can be found in Appendix~\ref{sec:ConvergenceAnalysis}.
\end{proof}
It is important to note that, in Theorem~\ref{thm:ConvergenceAnalysis}, $\omega_{t+1}$ is a virtual sequence that aligns with the DKN at the central server when $t+1\in\{nE\mid n=1,2,...\}$. The convergence analysis reveals that both the average and worst-case error variances play a critical role in determining the convergence behavior, highlighting the importance of uniformly reliable state estimation across all clients.
\begin{remark}
    Performance degradation caused by erroneous data can be mitigated through various approaches at different stages of KF-FedKL. On the client side, state estimations with large covariance can be excluded from training, and expert-defined thresholds can be applied to remove unrealistic estimates. During training, techniques such as adaptive learning rate adjustment~\cite{song2015learning} and error-tolerant loss function~\cite{fang2022robust} can be employed to improve robustness. On the server side, aggregation weights can be adaptively adjusted based on the deviation of local models from the central server's model. In addition, techniques such as Byzantine-robust federated learning~\cite{pmlr-v80-yin18a} can be employed to mitigate the impact of clients with erroneous data. The implementation of these mitigation strategies is beyond the scope of this paper and will be considered as an important direction for future work.
\end{remark}

\section{Numerical Results}\label{sec:Numerical}
This section provides a comprehensive performance evaluation, supported by extensive numerical results.

\subsection{Experiment settings}\label{sec:SimulationConfig}
We adopt the Lorenz63 system~\cite{Lorenz1963} as the default nonlinear system the clients observe. With the system state denoted as $\bm{x} =[x_1,x_2,x_3]$, the evolution of the Lorenz63 system is governed by
\begin{equation}\label{eq:Lorenz63}
    \begin{split}
        \dot{x}_1 & = a(x_2-x_1),     \\
        \dot{x}_2 & = x_1(b-x_3)-x_2, \\
        \dot{x}_3 & = x_1x_2 - cx_3,
    \end{split}
\end{equation}
where $a=10$, $b=28$, and $c=\frac{8}{3}$. The trajectory is obtained by numerically solving the Ordinary Differential Equations (ODEs) using the RK45 algorithm with a randomly initial state. It is important to note that the calculated trajectory is not available to the clients.

We consider a system with $N=5$ clients that adopt the same parameters except for those in the observation mechanism described by $h_n$. For clarity, we use client $n$ as an illustrative example. In a typical time slot, there are two processes. First, client $n$ independently observes the nonlinear system with a success probability $p_n=0.7$, where a successful observation yields $\zeta=300$ consecutive observation data. The observation mechanism is described by
\begin{equation}
    \bm{z}_{k,n} = h_n(\bm{x}_{k}) + \bm{v}_n = (\bm{x}_{k}-\bm{o}_n)\bm{b}_n^T + \bm{v}_n,
\end{equation}
where $\bm{b}_n \triangleq [\bm{b}_{n,1}; \bm{b}_{n,2}]$. We choose $\bm{b}_{n,1}$, $\bm{b}_{n,2}$, and $\bm{o}_n$ as three three-dimensional row vectors that are orthogonal to each other. Then, these observation data are utilized to estimate the system state using the Kalman filter introduced in Section~\ref{sec:Kalman}. For simulation purposes, we assume that the nonlinear system is restarted at the beginning of each time slot, and the initial state at each time slot is randomly and independently generated. The second process involves training the DKN. The central server determines which clients to activate based on a threshold-based policy. Specifically, the central server chooses client $n$ only when $|\mathcal{D}_n|\geq\rho$. Here, we set $\rho = 5\zeta$, which is equivalent to the data accumulated from five successful observations. When no client is activated, the system proceeds to the next time slot. Otherwise, each client receives the latest DKN from the central server and trains the received DKN by following the \textsc{ClientUpdate} function detailed in Algorithm~\ref{algo:MyFedAve}. Training is completed at the end of the current time slot, after which the updated DKN is uploaded to the central server, and the data used for training is discarded. To highlight key findings, we simulate the system for $M=200$ time slots.

In training, we set the number of epochs to $E=10$ with a batch size of $B=64$. The \texttt{Adam} optimizer~\cite{kingma2014adam} is adopted with learning rate $\eta=0.001$ and weight decay $\lambda=10^{-7}$. Additionally, we adopt an exponential learning rate scheduler with a multiplicative factor $\gamma=0.995$ to facilitate efficient training. We initialize the autoencoder in DKN independently for each client and the central server using He initialization with Kaiming normal distribution. Each element in the matrix $\bm{K}$ is randomly initialized following a normal distribution $\mathcal{N}(0,10^{-4})$. Table~\ref{tab:ParameterChoice} summarizes the default parameters used in numerical simulations. Unless stated otherwise, numerical results in the remainder of this paper are generated using these default parameters.
\begin{table}[t]
    \centering
    \begin{tabular}{lc|lc|lc} \toprule
        Parameter                             & Value    & Parameter     & Value  & Parameter  & Value         \\ \midrule
        $N$                                   & $5$      & $\alpha$      & $0.1$  & $E$        & $10$          \\
        $p_n$                                 & $0.7$    & $\kappa$      & $-1$   & $B$        & $64$          \\
        $\zeta$                               & $300$    & $\beta$       & $2$    & $\eta$     & $0.001$       \\
        $M$                                   & $200$    & $\mu=\tau$    & $1$    & $\lambda$  & $10^{-7}$     \\
        $\bm{\Sigma}_{f,n}=\bm{\Sigma}_{h,n}$ & $\bm{I}$ & $\mathcal{O}$ & $4d_x$ & $\gamma$   & $0.995$       \\
        $\hat{\bm{x}}_0$                      & $\bm{0}$ & $\mathcal{Z}$ & $1$    & $\rho_n$   & $5\zeta$      \\
        $\hat{\bm{P}}_0$                      & $\bm{I}$ & $\mathcal{H}$ & $30$   & $\omega_n$ & $\frac{1}{3}$ \\ \bottomrule
    \end{tabular}
    \caption{Summary of key parameters of KF-FedKL and their assigned values in numerical simulations.}
    \label{tab:ParameterChoice}
\end{table}

The DKN at the central server represents the learned linear model, and its accuracy is evaluated by prediction accuracy. Specifically, we adopt a combination of the following two measures. The first performance measure evaluates the DKN's ability to predict future states of the system. To this end, for $\zeta$ consecutive states, we define
\begin{equation}
    \mathscr{E}_{1,l} \triangleq \frac{1}{\zeta-l+1}\sum_{k=0}^{\zeta-l}\left(\frac{1}{d_x}||\bm{x}_{k+l}-\bar{\bm{x}}^{k+l}||_2^2\right),
\end{equation}
where $\bm{x}_k$ is the data used in performance evaluation and $\bar{\bm{x}}^{k+l}\triangleq\phi^{-1}(\bm{K}^{l}\phi(\bm{x}_k))$. The second performance measure evaluates the DKN's ability to predict future states in the Koopman-invariant subspace. Similarly, we define
\begin{equation}
    \mathscr{E}_{2,l} \triangleq \frac{1}{\zeta-l+1}\sum_{k=0}^{\zeta-l}\left(\frac{1}{d_y}||\bm{y}_{k+l}-\bm{K}^l\bm{y}_{k}||_2^2\right),
\end{equation}
where $\bm{y}_{k}\triangleq\phi(\bm{x}_{k})$ and $d_y$ is the dimensionality of $\bm{y}_{k}$. Then, the composite performance measure is defined as
\begin{equation}
    \mathscr{E} \triangleq \frac{1}{2l_1}\sum_{l=1}^{l_1}\mathscr{E}_{1,l} + \frac{1}{2l_2}\sum_{l=1}^{l_2}\mathscr{E}_{2,l}.
\end{equation}
In the simulation, we refer to $\mathscr{E}$ as the prediction error and choose $l_1=l_2=5$.

\subsection{Simulation Results}
We begin by evaluating the performance of \textsc{FedAvg-M} on the nonlinear systems listed in Table~\ref{tab:SystemConsidered}.
\begin{figure*}[t]
    \centering
    \begin{minipage}{0.295\textwidth}
        \includegraphics[width=\textwidth]{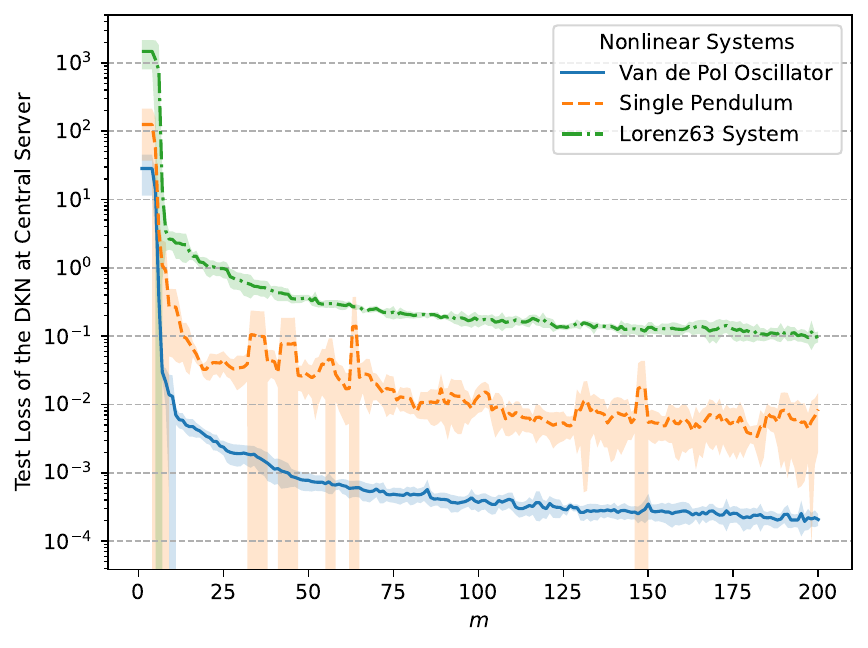}
        \subcaption{Canonical nonlinear systems.}
        \label{fig:System}
    \end{minipage}
    \begin{minipage}{0.32\textwidth}
        \includegraphics[width=\textwidth]{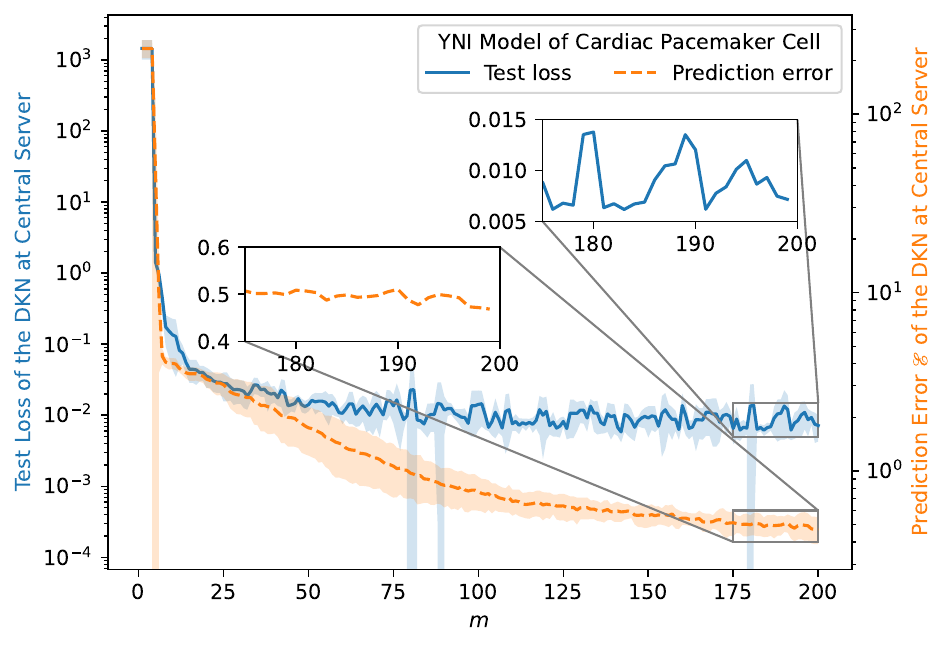}
        \subcaption{High-dimensional nonlinear system.}
        \label{fig:Application}
    \end{minipage}
    \begin{minipage}{0.32\textwidth}
        \includegraphics[width=\textwidth]{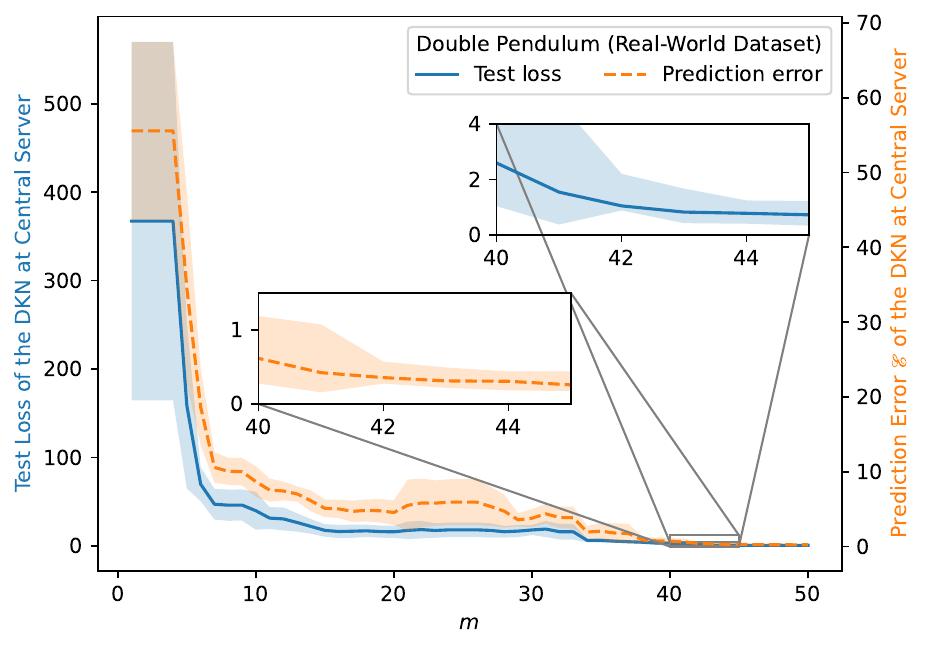}
        \subcaption{Nonlinear system with real-world dataset.}
        \label{fig:realdata}
    \end{minipage}
    \caption{Performance of KF-FedKL on various nonlinear systems. Lines represent the average of five independent simulations, and the shaded areas represent the standard deviations.}
    \label{fig:various_systems}
\end{figure*}
\begin{table*}[t]
    \centering
    \begin{tabular}{cccccc} \toprule
        {}                 & {System Name}                          & {Dimensionality}   & {Nonlinear Dynamics}                         & {Parameters}                                                  & Initial State                                      \\ \midrule
        \multirow{2}{*}{1} & \multirow{2}{*}{Van de Pol oscillator} & \multirow{2}{*}{2} & $\dot{x}_1 = x_2$                            & \multirow{2}{*}{$a=2$}                                        & \multirow{2}{*}{$x_i\in[-2,2]$ for $1\leq i\leq2$} \\
                           &                                        &                    & $\dot{x}_2 = a(1 - x_1^2)x_2-x_1$            &                                                               &                                                    \\\midrule
        \multirow{2}{*}{2} & \multirow{2}{*}{Single pendulum}       & \multirow{2}{*}{2} & $\dot{\theta} = \omega$                      & $L = 1$                                                       & $\theta\in[0,2\pi]$                                \\
                           &                                        &                    & $\dot{\omega} = -\frac{g}{L}\sin\theta$      & $g=9.8$                                                       & $\omega = 0$                                       \\\midrule
        3                  & Lorenz63 system                        & {3}                & Equation~\eqref{eq:Lorenz63}                 & Section~\ref{sec:SimulationConfig}                            & $x_i\in[-10,10]$ for $1\leq i\leq 3$               \\ \midrule
        \multirow{3}{*}{5} & \multirow{3}{*}{YNI model}             & \multirow{3}{*}{7} & $\dot{V} =-(I_{Na} + I_s + I_h + I_K + I_l)$ & \multirow{3}{*}{Details found in~\cite{doi2010computational}} & $V\in[-70,-50]$                                    \\
                           &                                        &                    & $\dot{x} = \alpha_{x}(V)(1-x)-\beta_x(V)x$   &                                                               & $x=0$                                              \\
                           &                                        &                    & $(x=m,h,d,f,q,p)$                            &                                                               & $(x=m,h,d,f,q,p)$                                  \\ \midrule
        \multirow{3}{*}{4} & \multirow{3}{*}{Double pendulum}       & \multirow{3}{*}{4} & $\dot{\theta}_1= \omega_1$                   & $L_1=0.172,L_2=0.143$                                         & \multirow{3}{*}{Not Applicable}                    \\
                           &                                        &                    & $\dot{\theta}_2 = \omega_2     $             & $m_1=0.311,m_2=0.111$                                         &                                                    \\
                           &                                        &                    & Equation~\eqref{eq:DoublePendulum}           & $g=9.8$                                                       &                                                    \\
        \bottomrule
    \end{tabular}
    \caption{Nonlinear systems for evaluating the performance of \textsc{FedAvg-M}.}
    \vspace{-1em}
    \label{tab:SystemConsidered}
\end{table*}
\begin{figure*}[b]
    \normalsize
    \vspace{-1em}
    \hrulefill
    \begin{equation}\label{eq:DoublePendulum}
        \begin{split}
            \dot{\omega}_1 = & \frac{L_1m_2\omega_1^2\sin(\theta_2-\theta_1)\cos(\theta_2-\theta_1)
                + gm_2\sin\theta_2\cos(\theta_2-\theta_1)
                + L_2m_2\omega_2^2\sin(\theta_2-\theta_1)
            - g(m_1+m_2)\sin\theta_1}{L_1(m_1+m_2) - L_1m_2\cos(\theta_2-\theta_1)\cos(\theta_2-\theta_1)} \\
            \dot{\omega}_2 = & \frac{-L_2m_2\omega_2^2\sin(\theta_2-\theta_1)\cos(\theta_2-\theta_1)
                + (m_1+m_2)(g\sin\theta_1\cos(\theta_2-\theta_1)
                - L_1\omega_1^2\sin(\theta_2-\theta_1)
                - g\sin\theta_2)}{L_2(m_1+m_2) - L_2m_2\cos(\theta_2-\theta_1)\cos(\theta_2-\theta_1)}
        \end{split}
    \end{equation}
    % \hrulefill
    % \vspace*{2pt}
\end{figure*}
First, Fig.~\ref{fig:System} presents results for three canonical nonlinear systems, specifically the first three listed in Table~\ref{tab:SystemConsidered}, with data obtained by numerically solving the corresponding ODEs. To eliminate the impact of Kalman filtering, we temporarily assume that each client has access to the system state $\bm{x}_k$ and uses these states in training. As observed, \textsc{FedAvg-M} effectively trains the DKNs at the central server for the nonlinear systems under consideration. Notably, the test loss increases as the dimensionality of the nonlinear system grows. Furthermore, when comparing the test losses for the Van der Pol oscillator and the single pendulum, we observe that the nonlinearity introduced by trigonometric functions presents greater challenges for linearization.

To assess the scalability of KF-FedKL to high-dimensional systems, we consider the Yanagihara-Noma-Irisawa (YNI) model of sinoatrial node cells, a classical Hodgkin-Huxley-type model of cardiac cells. The YNI model corresponds to the fourth nonlinear system in Table~\ref{tab:SystemConsidered}. For brevity, we omit the model details, which are available in~\cite{doi2010computational}. The dataset used in this experiment is generated by numerically solving the corresponding ODEs. We set the noise covariances as $\Sigma_{f,n} = \Sigma_{h,n} = 0.0025\bm{I}$ and use the identity observation function $h_n = \bm{I}$ for $1 \leq n \leq N$. The results are visualized in Fig.~\ref{fig:Application}, which shows that the training converges quickly and yields low prediction error, indicating that KF-FedKL can efficiently learn a linear model that captures the nonlinear dynamics in a high-dimensional system.

\begin{figure*}[t]
    \centering
    \begin{minipage}{0.33\textwidth}
        \includegraphics[width=\textwidth]{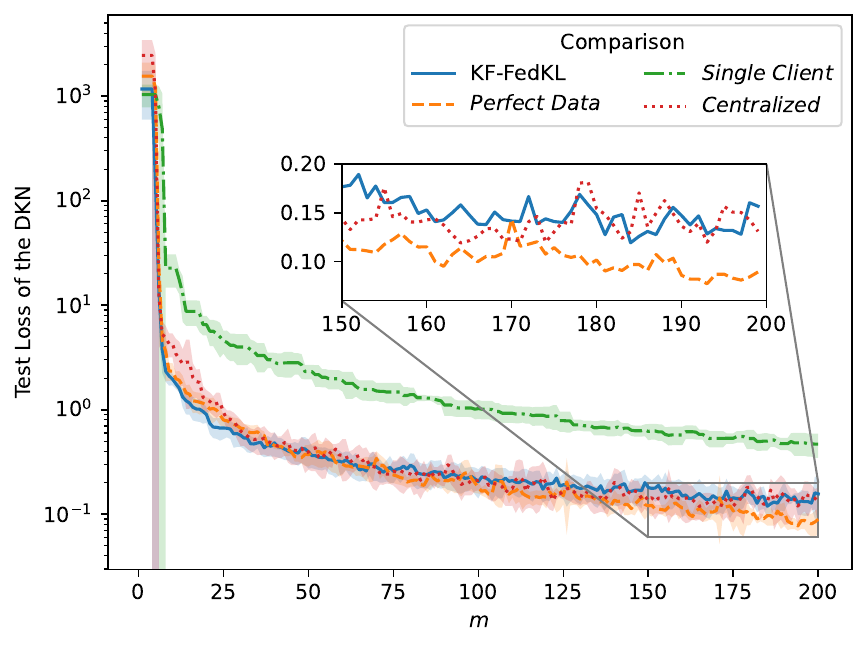}
        \subcaption{Test loss comparison.}
        \label{fig:comparison_loss}
    \end{minipage}
    \begin{minipage}{0.33\textwidth}
        \includegraphics[width=\textwidth]{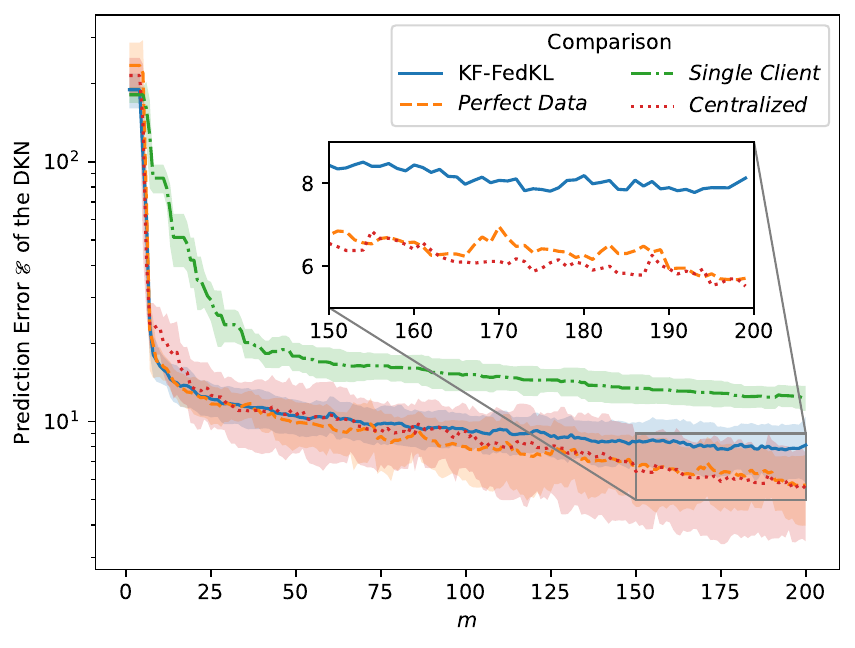}
        \subcaption{Prediction error comparison.}
        \label{fig:comparison_performance}
    \end{minipage}
    \begin{minipage}{0.285\textwidth}
        \includegraphics[width=\textwidth]{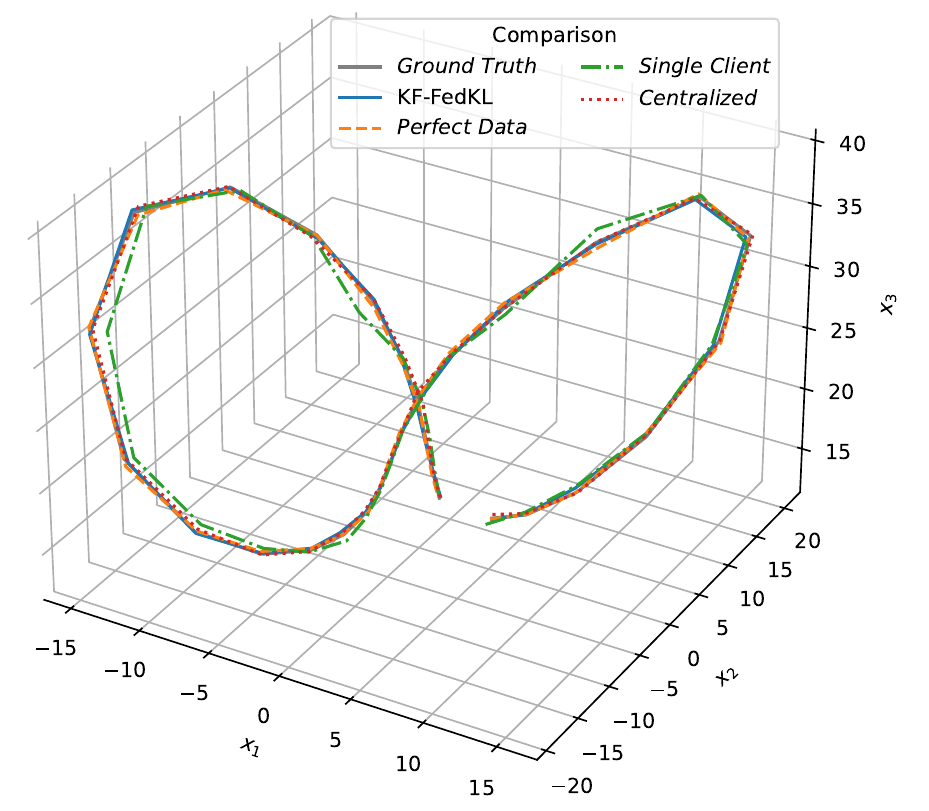}
        \subcaption{Predicted trajectory comparison.}
        \label{fig:comparison_trajectory}
    \end{minipage}
    \caption{Comparison of KF-FedKL with benchmark methods. In (a) and (b), lines represent the average of five independent simulations, and the shaded areas represent the standard deviations. In (c), we visualize trajectory segments predicted by the DKN using one-step predictions for each scheme, and compare them against the ground truth.}
    \vspace{-1em}
    \label{fig:framwork_compare}
\end{figure*}

To demonstrate the application of KF-FedKL on real-world systems, we use the dataset \texttt{Video\_Tracking\_Data.zip} provided in~\cite{myers2020low}, which contains real measured data from a double pendulum system. The dynamics and parameters of the double pendulum are specified as the fifth nonlinear system in Table~\ref{tab:SystemConsidered}. The dataset includes data from $3$ independent trials, each consisting of $38827$ samples recorded at a sampling rate of $500$Hz. However, we use only the samples indexed from $1000$ to $6000$, as the pendulums begin to move around sample $1000$ and the motion has largely subsided by sample $6000$. The excluded samples are noninformative and harmful to training according to our preliminary experiments. The dataset provides only the angular displacement $\theta_k$ of the two pendulums. To obtain the angular velocity $\omega_k$, we compute $\omega_k = \frac{\theta_{i_k} - \theta_{j_k}}{0.002(j_k-i_k)}$ where $i_k = \max\{0,k-200\}$ and $j_k=\min\{k+200,38827\}$. To accommodate the characteristics of the system and dataset, we make the following parameter choices. The observation function used by client $i$ is defined as $h_i = \bm{I}_{4 \times 4} + \Xi_i$, where $\Xi_i$ has the same dimensions as $\bm{I}_{4 \times 4}$, and each element is chosen randomly over the interval $[0, 0.5]$. The noise covariances are set to $\bm{\Sigma}_{f,n} = \bm{\Sigma}_{h,n} = \operatorname{diag}([0.01, 1, 0.01, 1])$ to account for model inaccuracies due to factors such as air resistance and friction in the mechanical joints. We simulate the process for $50$ time slots, with each successful observation in a single time slot obtaining data corresponding to $100$ samples. Each local training includes $25$ epochs. All other parameters follow the settings provided in Table~\ref{tab:ParameterChoice}. Samples from the first two trials are used for training. These samples are treated as system states in~\eqref{eq:noisy_state} prior to the observation and are not accessible to the clients. Samples indexed from $1000$ to $3000$ in the third trial are used for testing. The results are presented in Fig.~\ref{fig:realdata}. As shown, KF-FedKL effectively learns the nonlinear dynamics of the underlying physical system using real measured data. In particular, both the test loss and the prediction error of the DKN at the central server converge to low levels, demonstrating the effectiveness of the KF-FedKL when applied to real-world physical systems.

In the following, we compare the test losses and prediction errors achieved by the DKN at the central server under KF-FedKL with those achieved using the following benchmark schemes.
\begin{enumerate}
    \item \textit{Single Client}: The client trains its model independently without collaborating with other clients. It follows the same training pattern as in KF-FedKL, where training begins once the data accumulates to a threshold, and the used data is discarded after training.
    \item \textit{Perfect Data}: The clients have access to perfect data instead of using the output of the Kalman Filter. For simulations, the training data is obtained by numerically solving the corresponding ODEs.
    \item \textit{Centralized}: The clients transmit their data to the central server instead of training their local models. Specifically, at each time slot, the client's data is transmitted once it accumulates to a threshold and is discarded after being sent to the central server. At the central server, the received data is used for training during each time slot and is discarded once the training is completed. If no data is received, the central server stays idle.
\end{enumerate}
The results are provided in Fig.~\ref{fig:framwork_compare}. As we can see from Fig.~\ref{fig:comparison_loss}, the test loss under the \textit{Single Client} scheme is larger than that under the other three schemes. This is because, for a single client, the data is limited, and as a result, the DKN fails to effectively capture the characteristics of the nonlinear system. In contrast, the other three schemes achieve similar levels of test losses, with the \textit{perfect data} scheme achieving the lowest. Therefore, we can conclude that while the error introduced by the Kalman filter does impact training, its effect is limited. Furthermore, KF-FedKL demonstrates performance comparable to a centralized training scheme in terms of test loss. From Fig.~\ref{fig:comparison_performance}, we can observe that the \textit{Single Client} scheme exhibits the poorest performance. Conversely, the \textit{perfect data} and \textit{centralized} schemes achieve the best performance in terms of prediction error. However, even with a similar test loss, KF-FedKL results in a higher prediction error. This is because the test loss evaluates the one-step prediction error, while $\mathscr{E}$ assesses the multistep error. In Fig.~\ref{fig:comparison_trajectory}, we visualize trajectory segments predicted by the DKN using one-step predictions for each scheme, and compare them against the ground truth. These results are consistent with those presented in Fig.~\ref{fig:comparison_loss} and Fig.~\ref{fig:comparison_performance}. Specifically, the \textit{Single Client} scheme shows the poorest performance, while KF-FedKL delivers prediction accuracy comparable to those of the \textit{perfect data} scheme and \textit{centralized} scheme, demonstrating its effectiveness.

Additionally, we compare the performance of the UKF adopted in KF-FedKL with that of the following state estimation methods to demonstrate its efficiency.
\begin{itemize}
    \item \textit{UKF with resampling}: A variant of the UKF in which a new set of sigma points is sampled from the results of the prediction step prior to the correction step. This approach is particularly effective in scenarios with large or non-additive processes and observation noise.
    \item \textit{Extended Kalman filter~\cite{chui2017kalman}}: A classical nonlinear extension of the Kalman filter, often used as a benchmark. While computationally efficient, the extended Kalman filter may perform poorly in highly nonlinear systems.
    \item \textit{Particle filter~\cite{djuric2003particle}}: A sequential Monte Carlo method that approximates the posterior distribution using a set of weighted particles. This method typically incurs a higher computational cost.
\end{itemize}
The evaluation is conducted using $10$ system trajectories, each initialized with a randomly and independently generated initial state. The results are reported in Table~\ref{tab:estimation_performance}.
\begin{table}[t]
    \centering
    \begin{tabular}{c c c}
        \toprule
        \multirow{2}{*}{\textbf{State Estimation Methods}} & \textbf{$\Sigma_{f} = \Sigma_{h} = \bm{I}$} & \textbf{$\Sigma_{f} = \Sigma_{h} = 4\bm{I}$} \\
                                                           & \textbf{$\mathcal{E}\pm\mathcal{V}$}        & \textbf{$\mathcal{E}\pm\mathcal{V}$}         \\
        \midrule
        UKF                                                & $\mathbf{0.6755}\pm 0.3397$                 & $1.3686\pm 0.6850$                           \\
        UKF with resampling                                & $0.7173\pm 0.3650$                          & $\mathbf{1.3653}\pm 0.7091$                  \\
        Extended Kalman filter                             & $2.1662\pm0.8122$                           & $3.0394\pm1.5896$                            \\
        Particle filter                                    & $0.6814\pm\mathbf{0.3241}$                  & $1.3843\pm\mathbf{0.6324}$                   \\
        \bottomrule
    \end{tabular}
    \caption{Performance of various state estimation methods.}
    \label{tab:estimation_performance}
    \vspace{-1em}
\end{table}
In the table, $\mathcal{E}$ denotes the estimation error averaged over the final $100$ time slots and the $10$ trajectories. The error of a single estimation is defined as $\frac{1}{d_x} \left|\bm{x} - \hat{\bm{x}} \right|_1$, where $\bm{x}$ is the true system state, $\hat{\bm{x}}$ is the estimated state, and $|\cdot|_1$ denotes the $\ell^1$ norm. The symbol $\mathcal{V}$ in the table represents the standard deviation of the estimation error across the $10$ trajectories, averaged over the final $100$ time slots. As shown, when the process and observation noise levels are low (i.e., $\Sigma_{f} = \Sigma_{h} = \bm{I}$), UKF achieves the lowest overall estimation error. In this case, the resampling step offers no additional benefit, and the \textit{particle filter} performs comparably but with a slightly lower variance. However, the computational cost of \textit{particle filter} is significantly higher than that of the other methods considered. The \textit{extended Kalman filter} demonstrates the worst performance. Under high noise conditions (i.e., $\Sigma_{f} = \Sigma_{h} = 4\bm{I}$), \textit{UKF with resampling} slightly outperforms the other methods, though the improvement over the UKF is marginal. The \textit{particle filter} remains competitive, while the \textit{extended Kalman filter} continues to yield the poorest results. These findings support the selection of UKF as an effective and computationally efficient solution under the problem setting considered in this paper.

Finally, we investigate the impact of several key system parameters on the performance of KF-FedKL. We begin with an ablation study to highlight the critical role of each loss function used in~\eqref{eq:CompositeLoss}. To this end, we evaluate various combinations of the loss functions and perform five independent simulations for each setting. The results are presented in Table~\ref{tab:ablation_study}.
\begin{table}[t]
    \centering
    \begin{tabular}{cccc}
        \toprule
        \textbf{$\left[w_1,w_2,w_3\right]$}                & \textbf{$\bar{\ell}_1$} & \textbf{$\bar{\ell}_2$} & \textbf{$\bar{\ell}_3$} \\
        \midrule
        $\left[\frac{1}{3},\frac{1}{3},\frac{1}{3}\right]$ & $0.034$                 & $0.021$                 & $0.080$                 \\
        $\left[\frac{1}{2},\frac{1}{2},0\right]$           & $0.007$                 & $0.002$                 & $0.685$                 \\
        $\left[\frac{1}{2},0,\frac{1}{2}\right]$           & $0.027$                 & $1.530$                 & $0.086$                 \\
        $\left[0,\frac{1}{2},\frac{1}{2}\right]$           & $76.875$                & $0.014$                 & $0.078$                 \\
        $\left[0,0,1\right]$                               & $98.285$                & $1281.532$              & $0.080$                 \\
        $\left[0,1,0\right]$                               & $120.478$               & $0.002$                 & $274.563$               \\
        $\left[1,0,0\right]$                               & $8.270\times10^{-5}$    & $283.215$               & $283.957$               \\
        \bottomrule
    \end{tabular}
    \caption{Ablation study on loss function~\eqref{eq:CompositeLoss}.}
    \vspace{-1.5em}
    \label{tab:ablation_study}
\end{table}
In the table, $\bar{\ell}_i$ denotes the test loss corresponding to the $i$-th loss function, averaged over the final $50$ time slots and the $5$ independent simulations. As shown, excluding one or two loss functions leads to significantly higher test loss for the omitted ones, indicating that the DKN fails to fulfill the corresponding learning objectives. While it may appear that excluding certain loss functions yields improved performance on the included ones, this comes at the cost of overlooking important physical or structural aspects of the DKN. Therefore, despite marginal gains in selected loss functions, the omission of any loss function compromises the physical meaning of the DKN. This underscores the indispensability of all three loss functions for achieving a physically meaningful DKN.

\begin{figure*}[t]
    \centering
    \begin{minipage}{0.32\textwidth}
        \includegraphics[width=\textwidth]{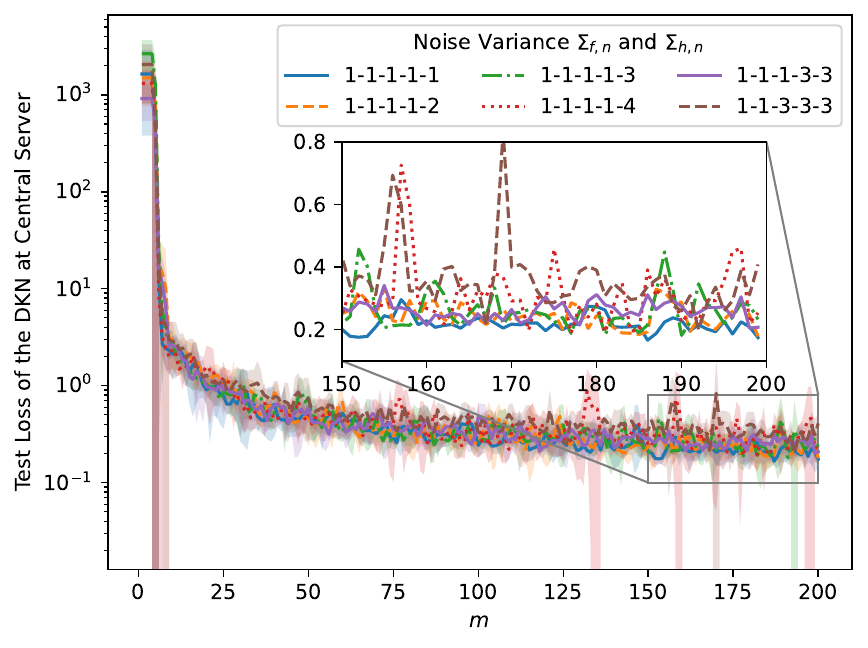}
        \subcaption{Noise variance $\Sigma_{f,n}$ and $\Sigma_{h,n}$.}
        \label{fig:Noise}
    \end{minipage}
    \begin{minipage}{0.32\textwidth}
        \includegraphics[width=\textwidth]{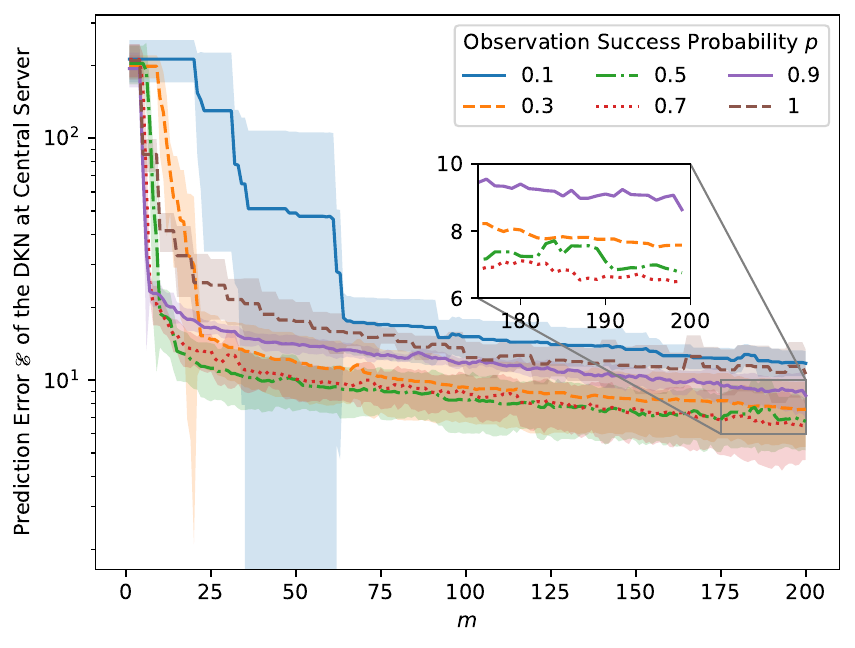}
        \subcaption{Observation success probability $p$.}
        \label{fig:Probability}
    \end{minipage}
    \begin{minipage}{0.32\textwidth}
        \includegraphics[width=\textwidth]{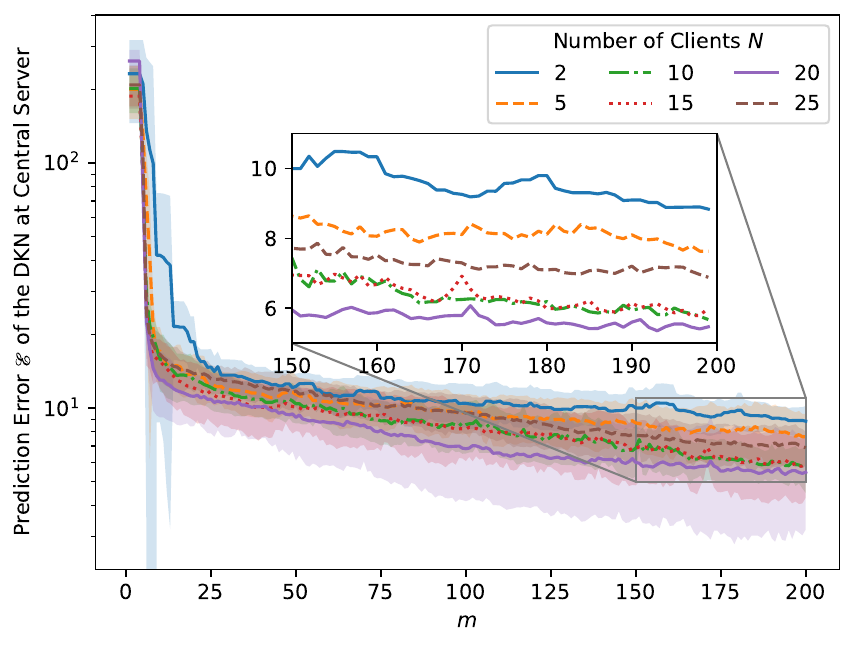}
        \subcaption{Client number $N$.}
        \label{fig:ClientNumber}
    \end{minipage}

    \vspace{0.5em}

    \begin{minipage}{0.32\textwidth}
        \includegraphics[width=\textwidth]{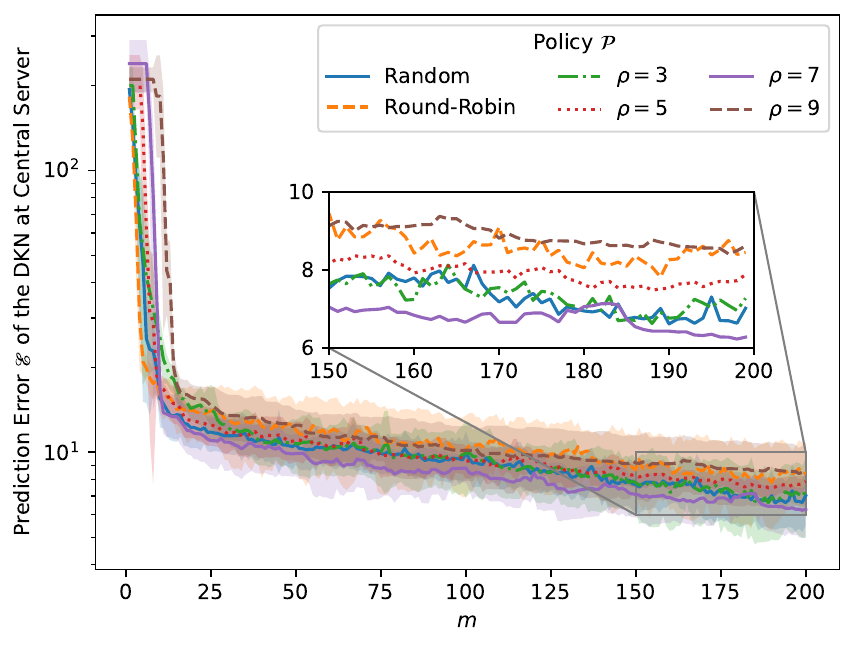}
        \subcaption{Policy $\mathcal{P}$.}
        \label{fig:Policy}
    \end{minipage}
    \begin{minipage}{0.32\textwidth}
        \includegraphics[width=\textwidth]{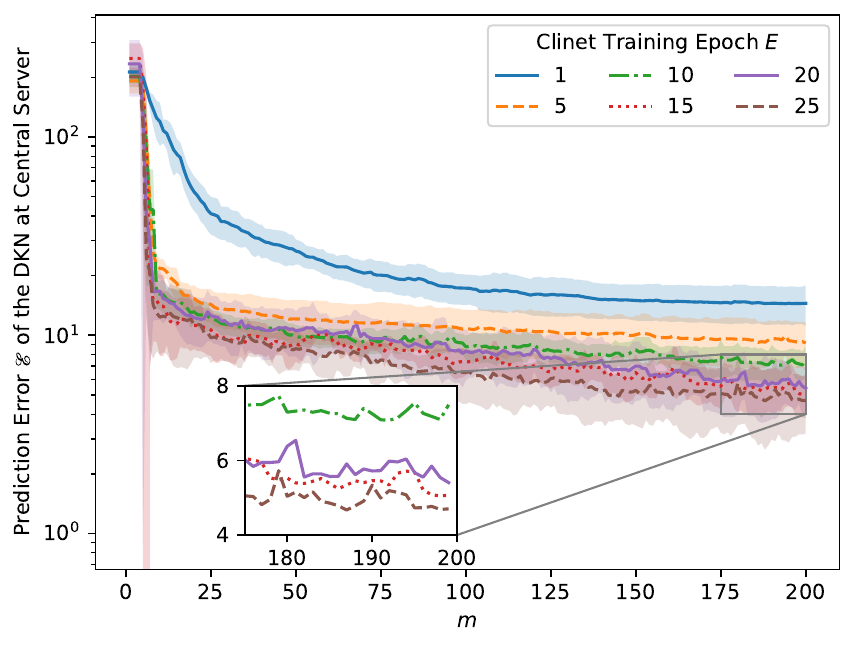}
        \subcaption{Training epoch $E$.}
        \label{fig:TrainingEpochs}
    \end{minipage}
    \begin{minipage}{0.32\textwidth}
        \includegraphics[width=\textwidth]{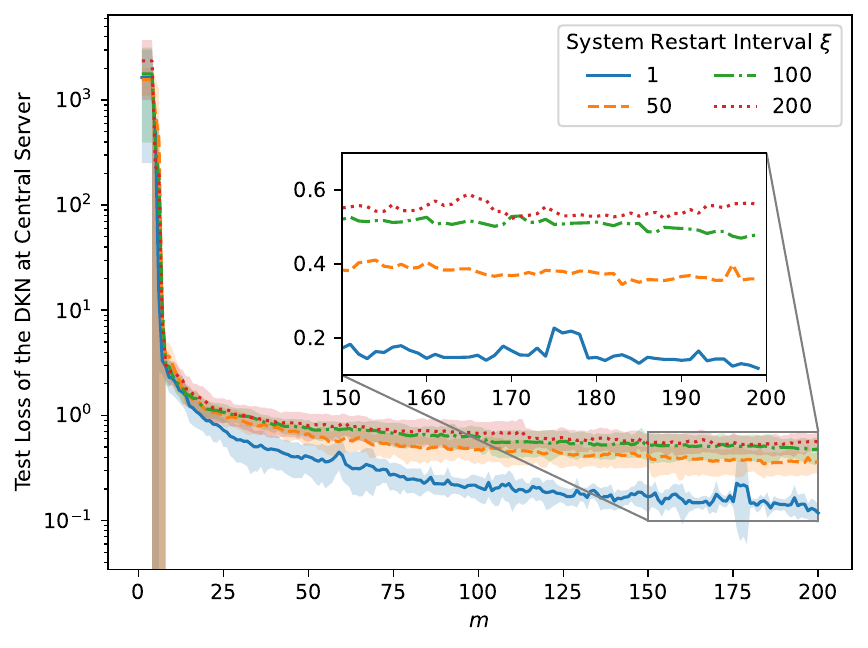}
        \subcaption{System restart interval $\mathcal{\xi}$.}
        \label{fig:Diversity}
    \end{minipage}
    \caption{The impact of system parameters on the performance. Lines represent the average of five independent simulations, and the shaded areas represent the standard deviations.}
    \vspace{-1em}
    \label{fig:parameters}
\end{figure*}
Then, we investigate the effect of the client's noise level $\Sigma_{f,n}$ and $\Sigma_{h,n}$ on the convergence. The results are visualized in Fig.~\ref{fig:Noise}, where the legend \textit{"1-1-1-1-3"} indicates four clients with $\Sigma_{f,n}=\Sigma_{h,n}=\bm{I}$ and one client with $\Sigma_{f,n}=\Sigma_{h,n}=3\bm{I}$. Note that the noise level can indirectly reflect the magnitude of errors in the estimated states. From the figure, we can see that higher noise or a greater number of noisy clients leads to increased fluctuations in test loss. This is because the excessive noise hinders the model's ability to accurately capture the nonlinear dynamics of the underlying system. Consequently, when such models are aggregated at the central server, they adversely impact the overall convergence. Nevertheless, this impact is mitigated to some extent by the averaging in federated learning, especially when aggregation is dominated by clients with lower noise levels. It is also important to note that in extreme scenarios where all clients experience high noise levels, the training fails to converge, which is not shown in the figure.

Fig.~\ref{fig:Probability} shows the performance when the clients adopt different probabilities $p$ of successful observation. In the simulation, we assume that all clients adopt the same probability. Specifically, $p_n = p$ for all $n$. The following observations can be made from the figure. The prediction error $\mathscr{E}$ first decreases and then increases as $p$ increases. The initial decrease occurs because the rate of data accumulation for clients becomes faster as $p$ increases. Therefore, the client participates in the training more often. At the same time, as $p$ reaches a certain value, the prediction error of the DKN at the central server increases. This is because the rapid accumulation of data by the clients and their frequent participation in training can potentially lead to overfitting. At the same time, the clients use estimated system states in training, whereas the test data is the true system states from numerical solutions. Consequently, the potential overfitting can lead to a DKN that is not generalized well, which in turn degrades the performance when tested on the true system state.

In Fig.~\ref{fig:ClientNumber}, we plot the performance of the DKN at the central server as the number of clients increases. The following observations can be made from the figure. The prediction error $\mathscr{E}$ first decreases and then increases as the number of clients increases. The initial decrease occurs because, under the threshold-based policy, more clients participate in each training round as the total number of clients increases. Additionally, when $N=25$, the prediction error increases. To explain this, we recall that the data used in training is the estimated system states. Hence, as the number of clients increases, the potential overfitting will result in a close fit between the output of the DKN and the estimated system states. Hence, the prediction error increases when the true system states from numerical solutions are used in the evaluation.

In Fig.~\ref{fig:Policy}, we investigate the performance resulting from the adoption of different policies.
\begin{definition}[Random policy]
    Under the Random policy, the central server randomly activates a single client in each time slot. Specifically, in time slot $m$ and for $1\leq n\leq N$,
    \begin{equation}
        \Pr\left\{\mathcal{I}_m=\{n\}\right\} = \frac{1}{N},
    \end{equation}
    where $\{n\}$ represents the set consisting of $n$ only.
\end{definition}
\begin{definition}[Round-Robin policy]
    Under the Round-Robin policy, the central server activates a single client in turn. Specifically, in time slot $m$,
    \begin{equation}
        \mathcal{I}_{m} = \{m\%N\},
    \end{equation}
    where $\%$ returns the remainder of the Euclidean division.
\end{definition}
\begin{definition}[Threshold policy $\rho$]
    Under Threshold policy $\rho$, the central server activates all the clients whose data size exceeds the threshold $\zeta\rho$. Specifically, in time slot $m$, $\mathcal{I}_{m} = \{n\mid |\mathcal{D}_n|\geq \zeta\rho\}$.
\end{definition}
\noindent We notice that there is little difference in the performance of the resulting DKN at the central server. Nevertheless, there are still notable phenomena. The threshold policy with $\rho=7$ yields the best performance. However, increasing or decreasing the threshold, which respectively incorporates more or fewer data in each training round, results in worse performance. The reason is that a higher threshold leads to more time slots where no clients are activated, which slows down the performance improvement. Conversely, decreasing the threshold results in less training data per training round, which slows down the generalization. It is important to note that many policies are not considered in this paper, such as the threshold policy with client-specific thresholds. Hence, identifying the optimal policy with respect to a specific performance measure is another optimization problem that requires further investigation.

In Fig.~\ref{fig:TrainingEpochs}, we examine the impact of the number of local training epochs. As shown, increasing the number of training epochs on the client side leads to improved performance. This is because more training epochs result in better generalization of the DKN. However, this benefit comes at the cost of increased training time. Several strategies can be employed to accelerate the training. These include adopting asynchronous federated learning~\cite{xu2023asynchronous} to address straggler issues, enhancing communication efficiency~\cite{konevcny2016federated}, and incorporating advanced optimization methods such as natural gradient descent~\cite{amari1998natural} and adaptive federated optimization~\cite{reddi2020adaptive}. In addition, system-level improvements, including parallel implementation, can further accelerate training. The implementation of these enhancements is an important direction for future work.

Motivated by the scenario where the real-world dataset exhibits limited diversity due to the inclusion of samples from only three independent trials, we further investigate the impact of data diversity on performance. To this end, we simulate a scenario in which the nonlinear system restarts every $\xi$ time slots, with larger values of $\xi$ corresponding to lower data diversity. As shown in Fig.~\ref{fig:Diversity}, the test loss consistently converges to a relatively small value across different settings. This is because the Lorenz63 system has similar trajectory shapes for different initial states, enabling the DKN to capture essential system dynamics even with reduced data diversity. However, when data becomes less diverse, test loss increases, as the reduced data diversity means a loss of information about the underlying system.

\section{Conclusion}
In this paper, we present a pioneering framework, Kalman Filter aided Federated Koopman Learning (KF-FedKL), which combines Kalman filter and federated learning with Koopman learning. KF-FedKL is designed for scenarios where multiple clients observe a common nonlinear system and collaborate for linearization purposes. To estimate the system state from the observation data, we leverage UKF and URTS smoother. Then, the clients collaboratively train the DKN under the federated learning framework using their estimated system states to achieve linearization. In training, we introduce a straightforward yet effective loss function to drive the training. We also study the case where the observation data arrives stochastically. To mitigate the effect of stochastic arrival, we propose the \textsc{FedAvg-M} algorithm. A convergence analysis is also presented to provide insights into the performance of KF-FedKL. Finally, extensive numerical results are presented to highlight the performance of KF-FedKL under various parameters.

\appendices
\section{Proof of Theorem~1}\label{sec:ConvergenceAnalysis}
Let $\mathcal{C}_E=\{nE\mid n = 1,2,...\}$ be the set of time slots where the communication happens. Then, we define
\begin{equation}
    \omega^n_t \triangleq \begin{cases}
        \nu^n_t                              & t\notin\mathcal{C}_E\ \text{or}\ |\mathcal{I}_t|=0, \\
        \sum_{n\in\mathcal{I}_t}q^n_t\nu_t^n & t\in\mathcal{C}_E\ \text{and}\ |\mathcal{I}_t|>0,
    \end{cases}
\end{equation}
where $q^n_t$ is the weights assigned to each client at time slot $t$ and, under the policy we adopted, $q^n_t = \frac{1}{|\mathcal{I}_t|}$ for $n\in\mathcal{I}_t$. We also define $\omega_t \triangleq \frac{1}{N}\sum_{n=1}^N\omega_t^n$ and $\nu_t \triangleq \frac{1}{N}\sum_{n=1}^N\nu_t^n$. Note that $\omega_t$ and $\nu_t$ are both virtual sequences as introduced in~\cite{FLConvergenceAnalysis0}. When $t\notin\mathcal{C}_E$, both are inaccessible. When $t\in\mathcal{C}_E$, only $\omega_t$ is accessible and is equivalent to the actual model at the central server. Then, we investigate
\begin{equation}\label{eq:ConAnaTarget}
    \begin{split}
        ||\omega_{t+1}-\omega^*||_2^2 = & ||\omega_{t+1} - \nu_{t+1} + \nu_{t+1} - \omega^*||_2^2           \\
        =                               & ||\omega_{t+1} - \nu_{t+1}||^2_2 + ||\nu_{t+1} - \omega^*||^2_2 + \\
                                        & 2\langle\omega_{t+1} - \nu_{t+1}, \nu_{t+1} - \omega^*\rangle,
    \end{split}
\end{equation}
where $\omega^*$ is the target model. In the following, we analyze each term in~\eqref{eq:ConAnaTarget} when the expectation is taken, resulting in the following lemmas.
\begin{lemma}\label{lem:term1}
    For $t\geq0$,
    \begin{equation}
        \mathbbm{E}\left[||\omega_{t+1} - \nu_{t+1}||_2^2\right] \leq \sum_{k=1}^NC_kE^2\eta_{t}^2\left(G^2+\sigma^2\right),
    \end{equation}
    where
    \begin{equation}
        C_k\triangleq\frac{4(N-k)N!p^k(1-p)^{N-k}}{k(N-1)k!(N-k)!}.
    \end{equation}
\end{lemma}
\begin{proof}
    When $t+1\notin\mathcal{C}_E$, $\mathbbm{E}\left[||\omega_{t+1} - \nu_{t+1}||_2^2\right] = 0$ since $\omega_{t+1} = \nu_{t+1}$ by definition. For the case of $t+1\in\mathcal{C}_E$, we have
    \begin{equation}
        \omega_{t+1} = \frac{1}{N}\sum_{n=1}^N\left(\sum_{n\in\mathcal{I}_{t+1}}q^n_{t+1}\nu_{t+1}^n\right) = \frac{1}{|\mathcal{I}_{t+1}|}\sum_{n\in\mathcal{I}_{t+1}}\nu_{t+1}^n.
    \end{equation}
    Then, we investigate the expected value of $||\omega_{t+1} - \nu_{t+1}||_2^2$ given that exactly $k>0$ clients are selected at time slot $k+1$. To this end, we have
    \begin{equation}
        \begin{split}
            A_k \triangleq & \mathbbm{E}\left[||\omega_{t+1} - \nu_{t+1}||_2^2 \,\middle|\, |\mathcal{I}_{t+1}|=k\right]                                                                                         \\
            =              & \frac{1}{k^2}\mathbbm{E}\left[\left\|\sum_{n=1}^N\mathbbm{1}_{\{n\in\mathcal{I}_{t+1}\}}\left(\nu_{t+1}^{n} - \nu_{t+1}\right)\right\|^2_2\,\middle|\, |\mathcal{I}_{t+1}|=k\right] \\
            =              & \frac{1}{k^2}\Bigg[\sum_{n=1}^NP_1||\nu_{t+1}^{n} - \nu_{t+1}||_2^2 +                                                                                                               \\
                           & \hspace{4em} \sum_{n\neq o}P_2\langle\nu_{t+1}^{n} - \nu_{t+1},\nu_{t+1}^{o} - \nu_{t+1}\rangle\Bigg],
        \end{split}
    \end{equation}
    where $\mathbbm{E}\left[B\,\middle|\, A\right]$ denotes the expected value of $B$ given $A$, and, for $1\leq n,o\leq N$,
    \begin{equation}
        P_1\triangleq \Pr\left\{n\in\mathcal{I}_{t+1}\right\} = \frac{k}{N},
    \end{equation}
    \begin{equation}
        P_2\triangleq \Pr\left\{n\in\mathcal{I}_{t+1},o\in\mathcal{I}_{t+1}\right\} = \frac{k(k-1)}{N(N-1)}.
    \end{equation}
    We notice that the following equality holds.
    \begin{equation}
        \sum_{n=1}^N||\nu_{t+1}^{n} - \nu_{t+1}||_2^2 + \sum_{n\neq o}\langle\nu_{t+1}^{n} - \nu_{t+1},\nu_{t+1}^{o} - \nu_{t+1}\rangle = 0
    \end{equation}
    Then, combining with $P_1$ and $P_2$ yields
    \begin{equation}
        A_k = \frac{N-k}{k(N-1)}\sum_{n=1}^N\frac{1}{N}||\nu_{t+1}^{n} - \nu_{t+1}||_2^2.
    \end{equation}
    In the following, we bound the term $\sum_{n=1}^N\frac{1}{N}||\nu_{t+1}^{n} - \nu_{t+1}||_2^2$ when the expectation is taken over the randomness of stochastic gradient~\cite{wei2022federated}. To this end, we have
    \begin{equation}
        \begin{split}
            \frac{1}{N}\mathbbm{E}\Bigg[\sum_{n=1}^N & ||\nu_{t+1}^{n} - \nu_{t+1}||_2^2\Bigg]                                                                        \\
            =                                        & \frac{1}{N}\sum_{n=1}^N\mathbbm{E}\left[||\nu_{t+1}^{n} - \omega_{t_0} - \nu_{t+1} + \omega_{t_0}||_2^2\right] \\
            \leq                                     & \frac{1}{N}\sum_{n=1}^N\mathbbm{E}\left[||\nu_{t+1}^n - \omega_{t_0}||_2^2\right],
        \end{split}
    \end{equation}
    where $t_0=\max\{t_0\mid t_0\in\mathcal{C}_E, t_0< t+1\}$, and we used the inequality $\mathbbm{E}\left[||x-\mathbbm{E}x||^2_2\right]\leq\mathbbm{E}\left[||x||_2^2\right]$. Then, we have
    \begin{equation}
        \begin{split}
            \sum_{n=1}^N\frac{1}{N}\mathbbm{E} & \left[||\nu_{t+1}^n - \omega_{t_0}||^2_2\right]                                                                                      \\
                                               & \leq \sum_{n=1}^N\frac{1}{N}E\sum_{i=t_0}^t\mathbbm{E}\left[||\eta_i\nabla\mathcal{L}_{i}^n(\omega_i^n;\xi_i^n) + e_i^n||_2^2\right] \\
                                               & \leq 4E^2\eta_{t}^2\left(G^2+\sigma^2\right),
        \end{split}
    \end{equation}
    where we used the fact that $\eta_t$ is decreasing, $\eta_{t_{0}}\leq2\eta_t$, and $\sigma^2=\max\{\sigma_n^2\}$. Combining, we have
    \begin{equation}\label{eq:lemma1}
        \mathbbm{E}\left[A_k\right] \leq \frac{4(N-k)}{k(N-1)}E^2\eta_{t}^2\left(G^2+\sigma^2\right).
    \end{equation}
    We notice that, when $k=0$, $A_0=0$ since $\omega_{t+1} = \nu_{t+1}$. Then, we have
    \begin{equation}
        \begin{split}
            \mathbbm{E} \left[||\omega_{t+1} - \nu_{t+1}||_2^2\right] = \sum_{k=1}^N\Pr\left\{|\mathcal{I}_{t+1}|=k\mid\mathcal{I}_t\right\}\mathbbm{E}\left[A_k\right],
        \end{split}
    \end{equation}
    where
    \begin{equation}
        \Pr\left\{|\mathcal{I}_{t+1}|=k\mid\mathcal{I}_t\right\} = \frac{N!p^k(1-p)^{N-k}}{k!(N-k)!}.
    \end{equation}
    Finally, combining with~\eqref{eq:lemma1} concludes the proof.
\end{proof}
\begin{lemma}\label{lem:term2}
    For $t\geq0$,
    \begin{equation}
        \begin{split}
            \mathbbm{E}\left[||\nu_{t+1} - \omega^*||_2^2\right]\leq & (1-\mu\eta_t)||\omega_t -\omega^*||_2^2 +                     \\
                                                                     & \eta_t^2\left[8(E-1)^2G^2 + 6L\Gamma\right] +                 \\
                                                                     & \frac{1}{N^2}\sum_{n=1}^N\left(\delta_n^2 + \sigma_n^2\right)
        \end{split}
    \end{equation}
\end{lemma}
\begin{proof}
    Under the equivalent scenario we considered, we can follow the analysis presented in~\cite{FLConvergenceAnalysis0}. To this end, we first define $g_t \triangleq \frac{1}{N}\sum_{n=1}^N\nabla\mathcal{L}_t^n(\omega_t^n;\xi_t^n)$, $\bar{g}_t \triangleq \frac{1}{N}\sum_{n=1}^N\nabla\mathcal{L}_t^n(\omega_t^n)$, and $e_t \triangleq \frac{1}{N}\sum_{n=1}^Ne_t^n$. Then, we have $\nu_{t+1} = \omega_t - \eta_t(g_t+e_t)$ and $\mathbbm{E}\left[g_t\right] = \bar{g}_t$. Consequently,
    \begin{equation}
        \begin{split}
            ||\nu_{t+1} - \omega^*||_2^2 = & ||\omega_t - \eta_tg_t- \eta_te_t -\omega^* - \eta_t\bar{g}_t + \eta_t\bar{g}_t||_2^2  \\
            =                              & ||\omega_t - \eta_t\bar{g}_t -\omega^*||_2^2 + \eta_t^2||g_t - \bar{g}_t + e_t||_2^2 + \\
                                           & 2\eta_t\langle\nu_t - \eta_t\bar{g}_t -\omega^*,\bar{g}_t -g_t - e_t\rangle
        \end{split}
    \end{equation}
    We investigate the expectation of each term separately. First,
    \begin{equation}\label{eq:lemma2-1}
        \mathbbm{E}\left[\langle\nu_t - \eta_t\bar{g}_t -\omega^*,\bar{g}_t -g_t - e_t\rangle\right] = 0,
    \end{equation}
    since $\mathbbm{E}\left[g_t\right] = \bar{g}_t$ and $\mathbbm{E}\left[e_t\right] =0$. Then,
    \begin{equation}\label{eq:lemma2-2}
        \begin{split}
            \mathbbm{E} & \left[||g_t - \bar{g}_t + e_t||^2_2\right]                                                                                                   \\
                        & = \frac{1}{N^2}\sum_{n=1}^N\mathbbm{E}\left[||\nabla\mathcal{L}_t^n(\omega_t^n;\xi_t^n)-\nabla\mathcal{L}_t^n(\omega_t^n)+e_t^n||^2_2\right] \\
                        & \leq \frac{1}{N^2}\sum_{n=1}^N\left(\delta_n^2 + \sigma_n^2\right).
        \end{split}
    \end{equation}
    Finally,
    \begin{equation}
        ||\omega_t - \eta_t\bar{g}_t -\omega^*||^2_2 = ||\omega_t -\omega^*||^2_2 - 2\eta_t\langle\omega_t -\omega^*,\bar{g}_t\rangle+ \eta_t^2||\bar{g}_t||^2_2.
    \end{equation}
    By following the same calculations in~\cite{FLConvergenceAnalysis0}, we can obtain
    \begin{equation}\label{eq:lemma2-3}
        \begin{split}
            \mathbbm{E}\left[||\omega_t - \eta_t\bar{g}_t -\omega^*||^2_2\right]\leq & (1-\mu\eta_t)||\omega_t -\omega^*||^2_2 + \\
                                                                                     & 8\eta_t^2(E-1)^2G^2 + 6\eta_t^2L\Gamma,
        \end{split}
    \end{equation}
    where
    \begin{equation}
        \Gamma \triangleq \mathcal{L}^* - \frac{1}{N}\sum_{n=1}^N\mathcal{L}^{n,*}.
    \end{equation}
    Combining~\eqref{eq:lemma2-1},~\eqref{eq:lemma2-2}, and~\eqref{eq:lemma2-3} concludes the proof.
\end{proof}
\begin{lemma}\label{lem:term3}
    For $t\geq0$,
    \begin{equation}
        \mathbbm{E}\left[\langle\omega_{t+1} - \nu_{t+1}, \nu_{t+1} - \omega^*\rangle\right] = 0.
    \end{equation}
\end{lemma}
\begin{proof}
    When $t+1\notin\mathcal{C}_E$, $\omega_{t+1} = \nu_{t+1}$ by definition. For the case of $t+1\in\mathcal{C}_E$, we first recall that the central server will select client independently with probability $p_n=p$ for $1\leq n\leq N$. When the central server selects exactly $k>0$ clients at time slot $t+1$, we have
    \begin{equation}
        \begin{split}
            \mathbbm{E}\big[\omega_{t+1} \,\big|\, & |\mathcal{I}_{t+1}|=k\big]                                                                                                                                  \\
                                                   & = \sum_{l=1}^{{N\choose k}}\Pr\left\{\mathcal{I}_{t+1} = \mathcal{I}^l_{t+1}\right\}\frac{\sum_{n\in\mathcal{I}^l_{t+1}}\nu_{t+1}^n}{|\mathcal{I}^l_{t+1}|} \\
                                                   & = \frac{1}{{N\choose k}k}\sum_{l=1}^{{N\choose k}}\sum_{n\in\mathcal{I}^l_{t+1}}\nu_{t+1}^n                                                                 \\
                                                   & = \frac{1}{N}\sum_{n=1}^{N}\nu_{t+1}^n = \nu_{t+1},
        \end{split}
    \end{equation}
    where $\mathcal{I}_{t+1}^l$ denotes the $l$-th possible set with $k$ clients. At the same time, $\mathbbm{E}\left[\omega_{t+1}\,\big|\, |\mathcal{I}_{t+1}|=0\right] = \nu_{t+1}$ since $\omega_{t+1} = \nu_{t+1}$ when no client is selected. Finally, when $t+1\in\mathcal{C}_E$,
    \begin{equation}
        \mathbbm{E}\left[\omega_{t+1}\right] \! = \!\sum_{k=0}^N\Pr\left\{|\mathcal{I}_{t+1}|=k\right\}\mathbbm{E}\left[\omega_{t+1}\,\big|\, |\mathcal{I}_{t+1}|=k\right] = \nu_{t+1}.
    \end{equation}
    Combining, we can conclude the proof.
\end{proof}
Combining Lemmas~\ref{lem:term1},~\ref{lem:term2}, and~\ref{lem:term3} concludes the proof.

\bibliographystyle{IEEEtran-no-space}
\bibliography{mybib}

\end{document}